\documentclass{article}
\usepackage{fullpage}
\usepackage{booktabs}
\usepackage{multirow}
\usepackage[utf8]{inputenc}
\usepackage{amsthm,amssymb}
\usepackage[leqno]{amsmath}
\usepackage{tcolorbox}
\usepackage{thmtools}
\usepackage{mathtools}
\usepackage{graphicx}
\usepackage{subcaption,thm-restate}
\usepackage[mathcal,mathscr]{eucal}
\usepackage{epsfig,graphicx,graphics,color}
\usepackage{caption}
\usepackage{enumerate}
\usepackage{enumitem}
\usepackage{xcolor}
\usepackage{titling}
\usepackage{hyperref}
\usepackage{verbatim}
\hypersetup{
	colorlinks=true,       
	linkcolor=blue,        
	citecolor=red,         
	filecolor=magenta,     
	urlcolor=cyan,         
	linktocpage=true
}
\usepackage{thm-restate}
\usepackage{cleveref}
\usepackage{float}
\usepackage{algpseudocode}
\usepackage{rotating}
\usepackage{makecell}
\usepackage{pdfpages}
\usepackage{graphicx}
\usepackage{svg}
\usepackage{tikz}
\usepackage{subcaption}
\usepackage{tabularx}
\usepackage{algorithm}
\usepackage{listings}

\usepackage[numbers, sort&compress]{natbib}

\theoremstyle{plain}
\newtheorem{theorem}{Theorem}
\newtheorem{lemma}[theorem]{Lemma}

\theoremstyle{definition}

\newtheorem{corollary}[theorem]{Corollary}

\usepackage{xcolor}
\definecolor[named]{ACMBlue}{cmyk}{1,0.1,0,0.1}
\definecolor[named]{ACMYellow}{cmyk}{0,0.16,1,0}
\definecolor[named]{ACMOrange}{cmyk}{0,0.42,1,0.01}
\definecolor[named]{ACMRed}{cmyk}{0,0.90,0.86,0}
\definecolor[named]{ACMLightBlue}{cmyk}{0.49,0.01,0,0}
\definecolor[named]{ACMGreen}{cmyk}{0.20,0,1,0.19}
\definecolor[named]{ACMPurple}{cmyk}{0.55,1,0,0.15}
\definecolor[named]{ACMDarkBlue}{cmyk}{1,0.58,0,0.21}
\definecolor[named]{ACMGrey}{rgb}{0.5, 0.5, 0.5} 
\definecolor{green}{rgb}{0,0.6,0}
\definecolor{fillblack}{rgb}{0.95,0.95,0.2195}

\def\final{0}  
\def\iflong{\iffalse}
\ifnum\final=0  
\newcommand{\mnote}[1]{{\color{orange}[{\tiny \textbf{Malte:} \bf #1}]\marginpar{\color{orange}*}}}
\newcommand{\jnote}[1]{{\color{red}[{\tiny \textbf{Julian:} \bf #1}]\marginpar{\color{red}*}}}
\newcommand{\fnote}[1]{{\color{blue}[{\tiny \textbf{Finn:} \bf #1}]\marginpar{\color{blue}*}}}
\else 
\newcommand{\mnote}[1]{}
\newcommand{\jnote}[1]{}
\newcommand{\fnote}[1]{}
\fi

\usepackage{authblk}

\title{The Line Traveling Salesman and Repairman Problem with Collaboration}

\date{}

\author{Julian Golak\thanks{Corresponding author}}
\author{Malte Fliedner}
\author{Finn Sörensen}

\affil{{\footnotesize Institute of Operations Management, University of Hamburg Business School, Hamburg, Germany. }}

\usepackage{setspace} 
\begin{document}
	\maketitle
	\begin{abstract}
		
		\noindent In this work, we consider extensions of both the Line Traveling Salesman and Line Traveling Repairman Problem, in which a single server must service a set of clients located along a line segment under the assumption that not only the server, but also the clients can move along the line and seek to collaborate with the server to speed up service times. 
        
		We analyze the structure of different problem versions and identify hard and easy subproblems by building up on prior results from the literature. Specifically, we investigate problem versions with zero or general processing times, clients that are either slower or faster than the server, as well as different time window restrictions. 
        Collectively, these results map out the complexity landscape of the Line Traveling Salesman and Repairman Problem with collaboration.
		\medskip
		
	\end{abstract}
		
	\section{Introduction and problem statement}
    Modern production and logistics systems are more and more characterized by an increased use of mobile robots that can autonomously move to their target locations in order to carry out designated tasks. If several of these robots are working together, either by supporting each other on specific tasks or by supplying each other with goods or tools, then new coordination problems arise that need to identify rendezvous positions and schedules in order to maximize the system's efficiency. This practical trend thus leads to relevant extensions of well-known optimization problems which can become research fields in their own right, as can for instance seen in the now famous Traveling Salesman Problem with sidekicks first introduced in \cite{MURRAY2015}. 

    In this work we seek to introduce the concept of collaboration among mobile units into the Linear Traveling Salesman (LTSP) and the Linear Traveling Repairman Problem (LTRP). Both are relevant subcases of the Traveling Salesman and Repairman Problem respectively, where all visited clients are located on a line. These cases are of theoretical interest because they reveal aspects of the combinatorial structure of routing problems when movement is simplified to a single dimension, but they are sometimes also of practical relevance, for instance in warehousing applications when movement is restricted to linear or semi-linear structures (e.g. see \cite{YANG2020}). Since warehousing is among those fields which has seen an increased use of mobile robots in recent years, we will study two classes of routing problem which can arise when a single robot (the server) has to collect items from a set of supplying robots \(A\) with \(n \coloneqq |A|\) (the clients) as efficiently as possible, while all movements occur along a line.
    
    For this purpose, each \emph{client} \(a\) is characterized by a \emph{time window}, defined by a \emph{release date} \(r(a) \in \mathbb{R}_+\) and a \emph{deadline} \(d(a) \in \mathbb{R}_+\). 
    At the release date, client \(a\in A\) is located at position \(s(a) \in \mathbb{R}\) and--without loss of generality--can either move along the line at speed \(v \in \mathbb{R}_+\) or remain at its position. 
    Additionally, each client has a specified \emph{processing time} \(\tau(a) \in \mathbb{R}_+\). 
    The \emph{server} starts at time \(0\) from the origin and can move along the line segment at unit speed or remain at its position. An input instance \(\mathcal{I}\) is given by the tuple \((A, v, (s(a), r(a), d(a), \tau(a))_{a \in A})\). 
    Further, we define \(L \coloneqq \{a \in A \mid s(a) \leq 0 \}\) and \(R \coloneqq \{a \in A \mid s(a) > 0 \}\). Let \(n_L \coloneqq |L|\) and \(n_R \coloneqq |R|\). 
	
	We aim to find \emph{rendezvous positions} between the server and the clients as well as a schedule which determines at what time the rendezvous is supposed to be carried out to maximize efficiency. 
    Define time \(t\colon A \to \mathbb{R}_+\) and position \(x \colon A \to \mathbb{R}\) such that the pair \((t(a), x(a))\) defines the rendezvous between the server and client \(a\in A\). A solution is thus represented by the pair \((t, x)\). 
	Given a solution \((t, x)\), we define the vectors \( \mathbf{t} = (t_0, t_1, \ldots, t_{n+1}) \) and \( \mathbf{x} = (x_0, x_1, \ldots, x_{n+1}) \), where for \(1 \leq i \leq n\), the pair \((t_i, x_i)\) represents the time and position of the server's \(i^{th}\) rendezvous. 
    The pair \((t_0, x_0)\) represents the server's starting time and position, while the pair \((t_{n+1}, x_{n+1})\) represents its ending time and position. 
    In the following, we will assume that the server starts and ends at the same position after processing all clients, requiring \(x_0 = x_{n+1} = 0\) for a feasible solution.
    Finally, we define sequence \( \mathbf{\sigma} = (\sigma_1, \ldots, \sigma_{n+1}) \), where \(\sigma_i\) represents client \(a\) that is processed in the \(i^{th}\) rendezvous. 
    Observe that the tuple \((\mathbf{\sigma}, \mathbf{t}, \mathbf{x})\) uniquely defines a solution \((t,x)\) and vice versa.
	
	When the server serves a client, it must remain at the rendezvous position for the duration of the processing time, during which the server cannot serve any other client. 
	For \(0 \leq i \leq n+1\), the completion time of the \(i^{\text{th}}\) rendezvous, denoted as \(c_i\), is defined as follows: \(c_i = 0\) if \(i = 0\); \(c_i = t_i + \tau(\sigma_i)\) for \(1 \leq i \leq n\); and \(c_i = t_{n+1}\) if \(i = n+1\).
	We refer to a rendezvous \(1 \leq i \leq n+1\) as \emph{reachable by the server} if \( t_i \geq c_{i-1} + |x_i - x_{i-1}| \). 
    Similarly, we refer to a rendezvous with client \(a \in A\) \emph{reachable by the client} if \( t(a) \geq r(a) + |x_i - s(a)|/v \).
	A solution \((t,x)\) is \emph{feasible} if the server meets every client in \(A\) and every meeting is reachable by both the server and the corresponding client and the server returns to its initial position and if \(r(a) \leq t(a) \leq d(a)\) for all \(a \in A\). 
	For a feasible solution \((t,x)\), the \emph{makespan} is \( C_{\text{max}}(t,x) = c_{n+1} \), and the sum of completion times is \( C_{\text{sum}}(t,x) = \sum^n_{i = 1} c_i \). 
    The objective is to find a feasible solution that minimizes either the makespan or the sum of completion times, which we refer to as an \emph{optimal} solution. 
    When the aim is to minimize the makespan, we refer to the problem as the \emph{Line Traveling Salesman Problem with collaboration (CLTSP)}. 
    Conversely, when the aim is to minimize the sum of completion times, we refer to it as the \emph{Line Traveling Repairman Problem with collaboration (CLTRP)}. 
    Additionally, we refer to the \emph{feasibility problem} as the problem of computing a feasible solution for both the CLTSP and CLTRP.
	
	Given a feasible solution, we define the \emph{trajectory} of the server as a piecewise linear curve in the time-space plane that intersects each point \( (t_i, x_i) \) for \( i = 0, \ldots, n+1 \). 
    We assume that, after completing a meeting, the server travels at unit speed to the next meeting position and waits if the client has not yet arrived. Similarly, clients travel at speed \(v\) to their meeting positions and wait if the server has not yet arrived. 
    Observe that the trajectories are uniquely defined.
	
	In the following, we define various restrictions on the problem’s assumptions. Depending on whether the clients or the server move faster, we distinguish two cases:
	\begin{itemize}
		\item[(A1)] \emph{Fast clients.} The clients move at a speed that is at least as fast as the server, i.e., \(v \geq 1\).
		\item[(A2)] \emph{Slow clients.} The clients move at a speed slower than the server, i.e., \(0 < v < 1\).
	\end{itemize} 
	Following classical scheduling literature, we distinguish between four cases based on how each end of the time window interval is bounded or unbounded:
	\begin{itemize}
		\item[(B1)] \emph{No time constraints.} The time windows of the clients impose no restrictions, i.e., \(r(a) = 0\) and \(d(a) = \infty\) for all \(a \in A\).
		\item[(B2)] \emph{Release times.} Only release times are considered, with no deadlines, i.e., \(d(a) = \infty\) for all \(a \in A\).
		\item[(B3)] \emph{Deadlines.} Only deadlines are considered, with no release times, i.e., \(r(a) = 0\) for all \(a \in A\).
		\item[(B4)] \emph{Time windows.} The most general case, where no assumptions are made about release times or deadlines.	
	\end{itemize}
	Finally, there are two cases regarding processing times for each client:
	\begin{itemize}
		\item[(C1)] \emph{Zero-processing times.} The service of a client is instantaneous, i.e., \(\tau(a) = 0\) for all \(a\in A\).
		\item[(C2)] \emph{General processing times.} The processing time of a client has no restrictions.
	\end{itemize}
	
    \paragraph{Previous work:}
    In terms of mobile clients, prior research has specifically focused on routing problems, where clients move on predefined and fixed trajectories, see e.g. \cite{helvig2003moving, stieber2015multiple, stieber2022dealing}. However, in our setting, the trajectories of the customers are subject to the decisions of the system operator which significantly changes the structure of the underlying optimization problems. 
    Surprisingly, the study of the routing problems with collaborative clients has received little attention thus far. To the best of our knowledge, Gambella, Naoum-Sawaya and Ghaddar \cite{Gambella2018} conducted the first comprehensive study on routing problems involving collaboration. In the context of ride-sharing, the authors formulated a vehicle routing problem in which a fleet of vehicles must pick up a set of customers. In this model, both vehicles and customers can move freely within a Euclidean plane. The aim is to find time and positions for rendezvous between vehicles and clients, such that the sum of total completion times is minimized. The authors provide exact solution methods based on decomposition techniques. 
    The research was further extended by Zhang et al.~\cite{Zhang2023}, who also investigated a vehicle routing problem with collaboration in the context of ride-sharing. The authors provided a more rigorous analysis by developing geometric solutions for specific cases with only a single vehicle and a single customer. 
    These solutions were then used to construct exact methods for solving the problem, and the authors derived managerial insights through extensive numerical experiments, utilizing both synthetic and real-world data.
	
    Due to their significance in the combinatorial optimization literature, the Traveling Salesman Problem and the Traveling Repairman Problem have both been studied extensively and various problem versions have been introduced and investigated. Since our work seeks to extend the research on those settings where all clients are located on a line, we focus on the literature that studies the LTSP and LTRP specifically in the following. All known results with respect to time-complexity are also summarized in Table~\ref{table:noCollab}.
    
    In the setting without time constraints, the LTSP becomes trivial, as an optimal solution involves the server moving first in one direction and then in the other.
    In the TRP, when there are no time constraints and zero-processing times, Afrati et al. \cite{afrati1986complexity} developed a simple quadratic-time algorithm based on a dynamic programming approach. This algorithm was later improved to a linear-time solution by Garcia, Jodra, and Tejel \cite{garcia2002note}. 
    When there are no time constraints and general processing time, the problem remains unsolved. 

    For the TSP and assuming only release time and zero-processing times, Psarafitis et al.~\cite{psaraftis1990routing} developed a quadratic-time algorithm.
    In the TRP and assuming the same restrictions, Sitters~\cite{sitters2004complexity} proved that the problem is at least binary NP-hard, yet there is no known pseudo-polynomial time algorithm.
    When considering release times and general processing times, Tsitsiklis \cite{tsitsiklis1992special} proved that the TSP is binary NP-hard, and the results in Lenstra, Kan and Brucker \cite{lenstra1977complexity} imply that the TRP is strongly NP-hard in this case.

    Under the assumption of only deadlines and zero-processing times, Garcia, Jodra, and Tejel \cite{garcia2002note} developed a linear-time algorithm for solving the feasibility problem. Tsitsiklis \cite{tsitsiklis1992special} proposed a quadratic-time algorithm for the TSP, while Afrati et al. \cite{afrati1986complexity} demonstrated that the TRP is binary NP-hard.
    Considering only deadlines and general processing times, Bock and Klamroth \cite{bock2013minimizing} proved that the TSP is binary NP-hard, while Bock~\cite{bock2015solving} further showed that the TRP is strongly NP-hard.
    Finally, in the case of time windows, Tsitsiklis \cite{tsitsiklis1992special} proved that the feasibility problem without processing times is strongly NP-hard, which implies the same level of hardness for both the TSP and TRP, with or without general processing times.
	
	\paragraph{Our Contribution:}
	
	With this article, we aim to introduce collaboration into the formal structure of two fundamental classes of routing problems and seek to lay the basis of a foundational complexity landscape for these problems. 
    In order to do this, we investigate different problem versions and either establish NP-completeness, present a polynomial-time algorithm or identify it as an open question. The results are summarized in Table~\ref{table:results}. 
    We further seek to highlight the differences between special cases with and without collaboration, emphasizing that certain problems proven to be hard in the non-collaborative setting become tractable with collaboration and vice versa. 
    Furthermore, some problem versions whose time complexity remains unresolved without collaboration can be classified by leveraging the collaborative structure.
	
	\medskip
	
	\noindent The remainder of the paper is structured as follows. In Section~\ref{sec:characterisingOpt}, we establish structural properties of optimal solutions across the problem variants. 
    In Section~\ref{sec:algorithmsTSP}, we present algorithms for the LTSP. 
    In Section~\ref{sec:algorithmsTRP}, we present algorithms for the LTRP. 
    Finally, we establish complexity lower bounds in Section~\ref{sec:hardness}.
	
	\paragraph{Basic Notations.} We denote the sets of \emph{real numbers}, \emph{non-negative real numbers}, \emph{integers}, and \emph{non-negative integers} by \(\mathbb{R}\), \(\mathbb{R}_+\), \(\mathbb{Z}\), and \(\mathbb{Z}_+\), respectively.
	
	\begin{table}[t!]
		\scriptsize
		\centering
		\renewcommand*{\arraystretch}{1.3}
		\setlength{\tabcolsep}{3pt}
		\newcommand{\tableentrybase}[2]{\colorbox{#1}{\parbox[c][3.5em][c]{18mm}{\centering\scriptsize #2}}}
		\newcommand{\tableentrybasee}[2]{\colorbox{#1}{\parbox[c][1.5em][c]{18mm}{\centering\scriptsize #2}}}
		\newcommand{\tableentrybaseee}[2]{\colorbox{#1}{\parbox[c][3.5em][c]{19mm}{\centering\scriptsize #2}}}
		\newcommand{\tableentrybaseeee}[2]{\colorbox{#1}{\parbox[c][1.5em][c]{10mm}{\centering\scriptsize #2}}}
		\newcommand{\tableentrybaseeeee}[2]{\colorbox{#1}{\parbox[c][3.5em][c]{10mm}{\centering\scriptsize #2}}}
		
		\caption{Complexity landscape of LTSP and LTRP without collaboration under time restrictions. In each column and row, complexities of cells with light shade are implied by cells with darker shade. \label{table:noCollab}}
		\begin{tabular}{cc|ccc|ccc}
			&  & \multicolumn{3}{c}{No Processing Time} & \multicolumn{3}{c}{General Processing Time}\\ 
			&  & Feasibility & TSP & TRP & Feasibility & TSP & TRP\\ \cmidrule(lr){1-2}\cmidrule(lr){3-3}\cmidrule(lr){4-4}\cmidrule(lr){5-5}\cmidrule(lr){6-6}\cmidrule(lr){7-7}\cmidrule(lr){8-8}\\[-1em]
			\rotatebox[origin=c]{90}{No Time Constraints} & & 
			\tableentrybaseee{ACMGreen!0}{\tableentrybase{ACMGrey!10}{Trivial}} & 
			\tableentrybaseee{ACMGreen!0}{\tableentrybase{ACMGrey!10}{Trivial}} & 
			\tableentrybaseee{ACMRed!0}{\tableentrybase{ACMGreen!50}{\(O(n) \) \cite{garcia2002note}}} & 
			\tableentrybaseee{ACMGreen!0}{\tableentrybase{ACMGrey!10}{Trivial}} & 
			\tableentrybaseee{ACMGreen!0}{\tableentrybase{ACMGrey!10}{Trivial}}& 
			\tableentrybaseee{ACMRed!0}{\tableentrybase{ACMGreen!00}{Open}} \\ [3em]
			
			\rotatebox[origin=c]{90}{Release Times} & & \tableentrybaseee{ACMGreen!0}{\tableentrybase{ACMGrey!10}{Trivial}} & \tableentrybaseee{ACMRed!0}{\tableentrybase{ACMGreen!50}{\(O(n^2)\) \cite{psaraftis1990routing}}} &  \tableentrybaseee{ACMRed!0}{\tableentrybase{ACMRed!50}{Binary NP-hard \cite{sitters2004complexity}}} & \tableentrybaseee{ACMGreen!0}{\tableentrybase{ACMGrey!10}{Trivial}} & \tableentrybaseee{ACMRed!0}{\tableentrybase{ACMRed!50}{Binary NP-hard \cite{tsitsiklis1992special}}}&  \tableentrybaseee{ACMRed!0}{\tableentrybase{ACMRed!50}{Strongly NP-hard \cite{lenstra1977complexity}}}
			\\ [3em]
			
			\rotatebox[origin=c]{90}{Deadlines} && 
			\tableentrybaseee{ACMRed!0}{\tableentrybase{ACMGreen!50}{\(O(n)\) \cite{garcia2002note}}} & 
			\tableentrybaseee{ACMRed!0}{\tableentrybase{ACMGreen!50}{\(O(n^2)\) \cite{tsitsiklis1992special}}} & 
			\tableentrybaseee{ACMGreen!0}{\tableentrybase{ACMRed!50}{Binary NP-hard \cite{afrati1986complexity}}} & 
			\tableentrybaseee{ACMGreen!0}{\tableentrybase{ACMRed!0}{Open}} & 
			\tableentrybaseee{ACMGreen!0}{\tableentrybase{ACMRed!50}{Binary NP-hard \cite{bock2013minimizing}}} & 
			\tableentrybaseee{ACMGreen!0}{\tableentrybase{ACMRed!50}{Strongly NP-hard \cite{bock2015solving}}} \\ [3em]
			
			\rotatebox[origin=c]{90}{Time Windows} && 
			\tableentrybaseee{ACMGreen!0}{\tableentrybase{ACMRed!50}{Strongly NP-hard \cite{tsitsiklis1992special}}} & 
			\tableentrybaseee{ACMGreen!0}{\tableentrybase{ACMRed!20}{Strongly NP-hard}} &  
			\tableentrybaseee{ACMGreen!0}{\tableentrybase{ACMRed!20}{Strongly NP-hard}}  & 
			\tableentrybaseee{ACMGreen!0}{\tableentrybase{ACMRed!20}{Strongly NP-hard}} & 
			\tableentrybaseee{ACMGreen!0}{\tableentrybase{ACMRed!20}{Strongly NP-hard}} & 
			\tableentrybaseee{ACMGreen!0}{\tableentrybase{ACMRed!20}{Strongly NP-hard}}
		\end{tabular}
	\end{table}
	
	\begin{table}[t!]
		\scriptsize
		\centering
		\renewcommand*{\arraystretch}{1.3}
		\setlength{\tabcolsep}{3pt}
		\newcommand{\tableentrybase}[2]{\colorbox{#1}{\parbox[c][3.5em][c]{18mm}{\centering\scriptsize #2}}}
		\newcommand{\tableentrybasee}[2]{\colorbox{#1}{\parbox[c][1.5em][c]{18mm}{\centering\scriptsize #2}}}
		\newcommand{\tableentrybaseee}[2]{\colorbox{#1}{\parbox[c][3.5em][c]{19mm}{\centering\scriptsize #2}}}
		\newcommand{\tableentrybaseeee}[2]{\colorbox{#1}{\parbox[c][1.5em][c]{10mm}{\centering\scriptsize #2}}}
		\newcommand{\tableentrybaseeeee}[2]{\colorbox{#1}{\parbox[c][3.5em][c]{10mm}{\centering\scriptsize #2}}}
		
		\caption{Complexity landscape of CLTSP and CLTRP. In each column and row, complexities of cells with light shade are implied by cells with darker shade.  The constant \(T\) is defined later and represents a value that is bounded by a linear function of the input size under unary encoding. \label{table:results}}
		\begin{tabular}{cc|ccc|ccc}
			& Speed & \multicolumn{3}{c}{No Processing Time} & \multicolumn{3}{c}{General Processing Time}\\ 
			&  & Feasibility & TSP & TRP & Feasibility & TSP & TRP\\ \cmidrule(lr){1-2}\cmidrule(lr){3-3}\cmidrule(lr){4-4}\cmidrule(lr){5-5}\cmidrule(lr){6-6}\cmidrule(lr){7-7}\cmidrule(lr){8-8}\\[-1em]
			\rotatebox[origin=c]{90}{  No Time Constraints} 
			& \tableentrybaseeeee{ACMRed!0}{
				\tableentrybaseeee{ACMRed!0}{Slow} \linebreak 
				\tableentrybaseeee{ACMRed!0}{Fast}
			} 
			& \tableentrybaseee{ACMRed!0}{
				\tableentrybase{ACMGrey!10}{Trivial}
			} 
			& \tableentrybaseee{ACMRed!0}{
				\tableentrybasee{ACMGreen!50}{\(O(n)\) \linebreak Thm. \ref{thm:TSPNoProcSlow}} \linebreak 
				\tableentrybasee{ACMGreen!20}{\(O(n )\)}
			} 
			& \tableentrybaseee{ACMRed!0}{
				\tableentrybasee{ACMGreen!20}{\(O((Tn)^2)\)} \linebreak 
				\tableentrybasee{ACMGreen!0}{Open}
			}
			& \tableentrybaseee{ACMRed!0}{
				\tableentrybase{ACMGrey!10}{Trivial}
			} 
			&\tableentrybaseee{ACMRed!0}{
				\tableentrybasee{ACMGreen!50}{\(O(n^3)\) \linebreak Thm. \ref{thm:TSPGenProcUnconstrained}} \linebreak 
				\tableentrybasee{ACMGreen!20}{\(O(n)\)}
			} 
			& \tableentrybaseee{ACMRed!0}{
				\tableentrybase{ACMGreen!00}{Open}
			} \\[3em]
			
			\rotatebox[origin=c]{90}{Release Times} & \tableentrybaseeeee{ACMRed!0}{\tableentrybaseeee{ACMRed!0}{Slow} \linebreak \tableentrybaseeee{ACMRed!0}{Fast}}& \tableentrybaseee{ACMGreen!0}{\tableentrybase{ACMGrey!10}{Trivial}} & \tableentrybaseee{ACMRed!0}{\tableentrybasee{ACMGrey!0}{Open}\linebreak \tableentrybasee{ACMGreen!20}{\(O(n)\)}} &  \tableentrybaseee{ACMGreen!0}{{\tableentrybasee{ACMRed!50}{Bin. NP-hard 
            Cor.~\ref{cor:CLTRP}
            }\linebreak \tableentrybasee{ACMBlue!0}{Open}}} & \tableentrybaseee{ACMGreen!0}{\tableentrybase{ACMGrey!10}{Trivial}} & \tableentrybaseee{ACMRed!0}{\tableentrybasee{ACMRed!50}{Bin. NP-hard 
            Cor.~\ref{cor:CLTSP}
            }\linebreak \tableentrybasee{ACMGreen!50}{\(O(n)\) \linebreak Thm. \ref{thm:TSPGenReleaseFast}}}  &  \tableentrybaseee{ACMRed!0}{\tableentrybasee{ACMRed!20}{Bin. NP-hard}\linebreak \tableentrybasee{ACMBlue!0}{Open}}\\ [3em]
			
			\rotatebox[origin=c]{90}{Deadlines} & \tableentrybaseeeee{ACMRed!0}{\tableentrybaseeee{ACMRed!0}{Slow} \linebreak \tableentrybaseeee{ACMRed!0}{Fast}} & \tableentrybaseee{ACMGreen!20}{O(n)} & \tableentrybaseee{ACMRed!0}{\tableentrybasee{ACMGreen!50}{\(O(n^2)\) \linebreak Thm.~\ref{thm:TSPNoProcDeadlines}} \linebreak \tableentrybasee{ACMGreen!20}{\(O(n)\)}} & \tableentrybaseee{ACMGreen!0}{{\tableentrybasee{ACMGreen!50}{
            \(O((Tn)^2)\) Thm.~\ref{thm:TRP}}
            }\linebreak \tableentrybasee{ACMGreen!0}{Open}} & 
			\tableentrybaseee{ACMGreen!0}{\tableentrybase{ACMRed!50}{Strongly NP-hard \linebreak Thm.~\ref{thm:HardnessStrong}}}
			& \tableentrybaseee{ACMGreen!0}{\tableentrybase{ACMRed!20}{Strongly NP-hard}}&  \tableentrybaseee{ACMGreen!0}{\tableentrybase{ACMRed!20}{Strongly NP-hard}} \\ [3em]
			
			\rotatebox[origin=c]{90}{Time Windows} & \tableentrybaseeeee{ACMRed!0}{\tableentrybaseeee{ACMRed!0}{Slow} \linebreak \tableentrybaseeee{ACMRed!0}{Fast}} & \tableentrybaseee{ACMRed!0}{\tableentrybasee{ACMRed!50}{Strongly NP-hard \cite{tsitsiklis1992special}} \linebreak \tableentrybasee{ACMGreen!20}{\(O(n)\)}} & \tableentrybaseee{ACMRed!0}{\tableentrybasee{ACMRed!20}{Strongly NP-hard} \linebreak \tableentrybasee{ACMGreen!50}{\(O(n)\) \linebreak Thm.~\ref{thm:TSPnoProcTimeWindow}}} &  \tableentrybaseee{ACMRed!0}{\tableentrybasee{ACMRed!20}{Strongly NP-hard}\linebreak \tableentrybasee{ACMBlue!0}{Open}} & \tableentrybaseee{ACMGreen!0}{\tableentrybase{ACMRed!20}{Strongly NP-hard}} & \tableentrybaseee{ACMGreen!0}{\tableentrybase{ACMRed!20}{Strongly NP-hard}}& \tableentrybaseee{ACMGreen!0}{\tableentrybase{ACMRed!20}{Strongly NP-hard}}
		\end{tabular}
	\end{table}

	\section{Structural properties of the Line Traveling Salesman Problem with collaboration} \label{sec:characterisingOpt}
	
	We propose a number of structural results that will be useful for proving the correctness of the algorithms that follow. Specifically, we show that there exists an optimal solution that has a very specific structure.
	
	\begin{lemma}\label{lemma:orderPreservingNoProc}
		If there are zero-processing times and only deadlines and if the instance is feasible, there exists an optimal solution for the CLTSP, such that for any clients \(a, a'\in L\) or \(a, a'\in R\) with \( |s(a)| < |s(a')|\) we have \(t(a) < t(a')\). In other words, the solution is order-preserving.
	\end{lemma}
	\begin{proof}
		Assume the instance is feasible and consider an optimal solution \( (t,x) \) and assume it is not order-preserving. 
		Assume there exists a client \(a \in R\), such that there exists an \(a' \in R\) with \( s(a) < s(a')\) and \(t(a) > t(a')\).  
		Consider lines \(f(t) = s(a) - vt\) and \(g(t) = s(a') - vt\).
		For the rendezvous to be reachable by \(a'\), it must hold that \(g(t(a')) \leq x(a')\).
		But then the trajectory of the server must cross \(f(t)\) at some point \( (\bar{t},\bar{x})\), where \(\bar{t} < t(a')\) and \(\bar{x} < x(a')\).
  
		Define solution \( (t',x')\) with \((t'(a), x'(a)) = (\bar{t},\bar{x})\) and \((t'(a''), x'(a'')) = (t(a''), x(a''))\) for all other \(a'' \in A\).
		Observe that both solutions have the same trajectory and therefore all rendezvous are reachable by the server and the clients.
		Further, we know that \(t'(a) \leq t(a) \leq d(a)\) for all \(a\in A\) and thus  \( (t',x')\) is feasible.
		It follows that \( (t',x')\) is again an optimal solution and we know that \(t'(a') > t'(a)\). If there exists a client \(a \in L\), such that there exists an \(a' \in L\) with \( s(a) > s(a')\) and \(t(a) > t(a')\), a symmetric procedure is used. Repeating this procedure for all relevant elements, results in a solution that satisfies the condition stated in the lemma.
	\end{proof}
	
	\begin{lemma}\label{lemma:orderPreservingGenProc}
		If there are general processing times and no time constraints, then there exists an optimal solution for the CLTSP, such that for any clients \(a, a'\in L\) or \(a, a'\in R\) with \( |s(a)| < |s(a')|\) we have \(t(a) < t(a')\). In other words, the solution is order-preserving.
	\end{lemma}
	\begin{proof}
        Observe that there always exists an optimal solution.
		Consider an optimal solution \( (t,x) \). 
		Assume there exists a client \(a \in R\), such that there exists an \(a' \in R\) with \( s(a) < s(a')\) and \(t(a) > t(a')\).  
        If there do not exist such clients, we are done.
		Consider lines \(f(t) = s(a) - vt\) and \(g(t) = s(a') - vt\).
		But then the trajectory of the server must cross \(f(t)\) at some point \( (\bar{t},\bar{x})\), where \(\bar{t} < t(a')\) and \(\bar{x} < x(a')\).
  
		Define solution \( (t',x')\) by defining \((t'(a), x'(a)) = (\bar{t}, \bar{x})\) and for all \(a'' \in A \setminus \{a\}\) 
		\[
		(t'(a''), x'(a'')) =
		\begin{cases} 
			(t(a'') + \tau(a), x(a'')), & \text{if } t(a') < t(a'') < t(a), \\
			(t(a''), x(a'')), & \text{otherwise}.
		\end{cases}
		\]
		Observe that, by construction, all rendezvous are reachable by the server and the clients.
		It follows that \( (t',x')\) is again an optimal solution and we know that \(t'(a') > t'(a)\). If there exists a client \(a \in L\), such that there exists an \(a' \in L\) with \( s(a) > s(a')\) and \(t(a) > t(a')\), a symmetric procedure is used. Repeating this procedure for all relevant elements, results in a solution that satisfies the condition stated in the lemma.
	\end{proof}
	
	\begin{lemma}\label{lemma:waitFreeServer}
		If there are general processing times and only deadlines and clients are slow and if the instance is feasible, then given an optimal solution \((t, x)\) for the CLTSP, we have that \(t_{i} = c_{i-1} + |x_{i} - x_{i-1}|\) for all \(1 \leq i \leq n+1\). In other words, the server is wait-free. 
	\end{lemma}
	\begin{proof}
		The argument is illustrated in Figure~\ref{fig:proofWaitFree1}. Assume the instance is feasible. Consider an optimal solution \( (t, x)\) and assume that there exists an \( i \geq 1 \) such that \( t_i > c_{i-1} + |x_i - x_{i-1}| \). 
        If there does not exist such a rendezvous, we are done. Let \( a = \sigma_{i} \), and denote the position of \( a \) at time \( c_{i-1} \) as \( s'(a)\). We may assume that \(s'(a)\geq x_{i-1}\). Consider solution \((t',x')\) by defining \(\sigma' = \sigma\) and for all \(j\),
		\[
		(t'_j, x'_j) =
		\begin{cases} 
			\left(c_{i-1} + (s'(a) - x_{i-1})/(1+v), x_i + (s'(a) - x_{i-1})/(1+v)\right), & \text{if } j = i,\\
			(t_j, x_j), & \text{otherwise.}
		\end{cases}
		\]
		Observe that, by construction, every rendezvous is reachable by the clients. Further, we know that \(t'_i \leq t_i\) and \(x'_i \leq x_i\). 
        Since \((t_i, x_i)\) is reachable by \(a\), we know that \(c_i \geq t'_i + (x_i - x'_i)/{v} + \tau(a)\). Thus, in solution \((t',x')\) the server can process \(a\) at position \(x_i\) in time \( t'_i + x_i - x'_i + \tau(a) < t_i + (x_i - x'_i)/{v} + \tau(a) \leq c_i\) using that \(v < 1\) and \(x_i - x'_i \geq 0\) and \(t'_i \leq t_i\). This implies that every rendezvous is reachable by the server. 
		Further, we know that \(t'(a) \leq t(a) \leq d(a)\) for all \(a\in A\) and thus  \( (t',x')\) is feasible.
		It follows that \( (t',x')\) is again an optimal solution and we know that the rendezvous with \(a\) is wait-free. Repeating this procedure for all relevant elements results in a solution that satisfies the condition stated in the lemma.
		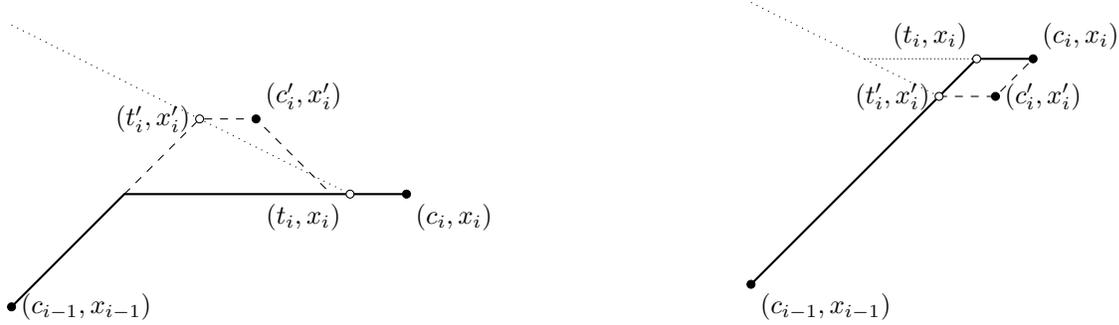
\begin{figure}[t!]
			\centering
			\begin{subfigure}[b]{0.45\textwidth}
				\centering
				\begin{tikzpicture}[scale=1.5]
                \coordinate (startLast) at (0,-1); 
                \coordinate (startClient) at (0,1.5); 
                \coordinate (startOpt) at (5/3,2/3);  
                \coordinate (endOpt) at (5/3 + 0.5,2/3);    
                \coordinate (backOpt) at (7/3+0.5, 0); 
                \coordinate (posWait) at (1,0); 
                \coordinate (startWait) at (3,0); 
                \coordinate (endWait) at (3.5,0);   
                
                \draw[dotted] (startClient) -- (startWait); 
                \draw[thick] (startLast) -- (posWait) -- (endWait); 
                \draw[dashed] (posWait) -- (startOpt) -- (endOpt) -- (backOpt);
                
                \filldraw[fill=black] (startLast) circle (1pt) node[right] {\((c_{i-1}, x_{i-1})\)}; 
                \filldraw[fill=white](startOpt) circle (1pt) node[left] {\((t_i', x_{i}')\)};
                \filldraw[fill=black] (endOpt) circle (1pt) node[above right] {$(c_i', x_i')$};
                \filldraw[fill=white] (startWait) circle (1pt) node[below left] {\((t_i, x_{i})\)};
                \filldraw[fill=black] (endWait) circle (1pt) node[below right] {\((c_i, x_i)\)};
                \end{tikzpicture}
                \caption{Illustration of the argument corresponding to Lemma~\ref{lemma:waitFreeServer}}
				\label{fig:proofWaitFree1}
			\end{subfigure}
			\hfill
			\begin{subfigure}[b]{0.45\textwidth}
				\centering
				\begin{tikzpicture}[scale=1.5]
                \coordinate (startLast) at (0,-1); 
                \coordinate (startClient) at (0,1.5); 
                \coordinate (startOpt) at (5/3,2/3);  
                \coordinate (endOpt) at (5/3 + 0.5,2/3);    
                \coordinate (posWait) at (1,1); 
                \coordinate (startWait) at (2,1); 
                \coordinate (endWait) at (2.5,1);   
                
                \draw[dotted] (startClient) -- (startOpt); 
                \draw[thick] (startLast) -- (startWait) -- (endWait); 
                \draw[dashed] (startOpt) -- (endOpt)  -- (endWait); 
                \draw[densely dotted] (posWait) -- (startWait); 
                
                \filldraw[fill=black] (startLast) circle (1pt) node[below right] {\((c_{i-1}, x_{i-1})\)}; 
                \filldraw[fill=white] (startOpt) circle (1pt) node[left] {\((t_i', x_{i}')\)};
                \filldraw[fill=black] (endOpt) circle (1pt) node[right] {$(c_i', x_i')$};
                \filldraw[fill=white] (startWait) circle (1pt) node[above left] {\((t_i, x_{i})\)};
                \filldraw[fill=black] (endWait) circle (1pt) node[above right] {\((c_i, x_i)\)};
                
                \end{tikzpicture}
                \caption{
                Illustration of the argument corresponding to Lemma~\ref{lemma:waitFreeClient}}
				\label{fig:proofWaitFree2}
			\end{subfigure}
			\caption{Illustration of the argument corresponding to Lemma~\ref{lemma:waitFreeServer} and \ref{lemma:waitFreeClient}. 
            Depicted are time-space diagrams, where the horizontal axis corresponds to time and the vertical axis corresponds to space.
            The solid line represents the trajectory of the server in solution \((t,x)\). 
            The dashed line represents the trajectory of the server in solution \((t',x')\).
            The dotted line represents the trajectory of the client in both solutions.
            The circles indicate the time-space pairs discussed in the argument, where the white circles represent the rendezvous time and the black circles the completion time.}
			\label{fig:proofWaitFreeCombined}
		\end{figure}
	\end{proof}
	
	\begin{lemma}\label{lemma:waitFreeClient}
		In the CLTSP, if the instance is feasible, there exists an optimal solution \((t,x)\) such that for all \(a\in R\), and let \(i\) be such that \(a = \sigma(i)\), we have that \( t(a) = \max\{r(a) + (s(a) - x(a))/v, c_{i-1}\} \), otherwise. For all \(a\in L\), and let \(i\) be such that \(a = \sigma(i)\), we have that \(t(a) = \max\{r(a) + (x(a) - s(a))/v, c_{i-1}\} \). In other words, the clients are colliding. 
	\end{lemma}
	\begin{proof}
		The argument is illustrated in Figure~\ref{fig:proofWaitFree2}.
        Assume the instance is feasible. 
        Consider an optimal solution \((t,x)\) such that there exists a client \(a\) that is not colliding.
        If there does not exist such client, we are done.
        We may assume that \(a \in R\) and 
		we may assume that the server is wait-free.
		Let \(a = \sigma_i\) and \(s'(a)\) denote the position of \(a\) at time \(c_{i-1}\). 
		Note that, by assumption, the client only waits upon arriving at their rendezvous position. 
        Therefore, we know that client a has not waited before time \(c_{i-1}\).
		We may assume that \(s'(a) \geq x_{i-1}\).
		Consider solution \((t',x')\) by defining \(\sigma' = \sigma\) and all \(j\),
		\[
		(t'_j, x'_j) =
		\begin{cases} 
			\left(c_{i-1} + (s'(a) - x_{i-1}) / (1+v), s'(a) - v(s'(a) - x_{i-1}) / (1+v)\right), & \text{if } j = i,\\
			(t_j, x_j), & \text{otherwise.}
		\end{cases}
		\]
		By construction all rendezvous are reachable by the clients.
		We know that both \((t_i,x_i)\) and \((t'_i,x'_i)\) are reached by the server without waiting, thus it holds that \(c_i = t'_i + x_i-x'_i + \tau(a) \). 
        Thus, in solution \((t',x')\) the server can process \(a\) at position \(x_i\) in time \(t'_i + x_i - x'_i + \tau(a)  \leq c_i\). 
        This implies that every rendezvous is reachable by the server. Further, we know that \(t'(a) \leq t(a) \leq d(a)\) for all \(a\in A\) and thus  \( (t',x')\) is feasible.
		It follows that \( (t',x')\) is again an optimal solution and we know that \(a\) is colliding. 
        Repeating this procedure for all relevant elements, results in a solution that satisfies the condition stated in the lemma.
	\end{proof}	
	\section{Algorithms for the Line Traveling Salesman Problem with collaboration
	}\label{sec:algorithmsTSP}
	
	In this section, we give an algorithm that computes an optimal solution for the CLTSP for a selection of variants. In Section~\ref{sec:linAlg}, we propose (log)-linear-time algorithms for some special cases. In Section~\ref{sec:DP} we present dynamic programming algorithms for more general cases.
	
	\subsection{Linear- and log-linear-time algorithms}\label{sec:linAlg}
	In this section, we present three special cases whose structure immediately implies to linear-time or log-linear-time algorithms.
	
		\begin{theorem}\label{thm:TSPNoProcSlow}
		If there are zero-processing times and no time constraints and if the clients are slow, then there exists an algorithm that solves the CLTSP in \(O(n)\) time.
	\end{theorem}
	\begin{proof}
		We claim that in an optimal solution the clients in sets \(L\) and \(R\) are processed consecutively. 
		Consider an optimal solution \((t,x)\), where this does not hold.
		By Lemmas~\ref{lemma:orderPreservingNoProc}, \ref{lemma:waitFreeServer} and \ref{lemma:waitFreeClient}, we may assume that \((t,x)\) is order-preserving, the server is wait-free and clients are colliding. 
		Let \(\ell\) represent the client in \(L\) with the greatest starting distance from the origin. Similarly, let \(r\) represent the client in \(R\) with the greatest distance from the origin.
		We may assume that \(r\) is the last client to be processed and that, prior to \(\ell\), the server processed clients from \(R\).
		Consider solution \((t',x')\) that corresponds to the trajectory, such that the server travels to \(\ell\), then to \(r\), then to back to the origin. 
		It is clear, that the trajectory of the server intersects with the trajectory of each client. 
        Since in \((t, x)\) the server processes clients in \(R\) before reaching \(\ell\), whereas in \((t', x')\) it travels directly to \(\ell\), it follows that \(t'(\ell) < t(\ell)\).
        In both solutions, the server travels from \(\ell\) directly to \(r\), implying that \(t'(r) < t(r)\). 
        As the clients are slow, we know that when the server intersects the trajectory of \(r\) earlier, it also arrives at the origin earlier. 
        Consequently, we have that \(C_{\text{max}}(t',x') < C_{\text{max}}(t,x)\), implying a contradiction.
		
		We may assume that \(s(r) > |s(\ell)|\). We claim that in an optimal solution the clients in \(R\) are processed first, followed by the clients in \(L\).
		Consider an optimal solution \((t,x)\), where this does not hold. 
		Again, we assume that \((t,x)\) is order-preserving, the server is wait-free and clients are colliding. 
		By the previous claim, we know that the server first processes clients in \(L\), then in \(R\). 
		Consider solution \((t',x')\), where the server first processes clients in \(R\), then in \(L\).
        The following argument is illustrated in Figure~\ref{fig:proofLinearSlow}.
		We know that the trajectories corresponding to both solutions are fully determined by their rendezvous with \(\ell\) and \(r\). 
		With straightforward algebra, we now determine these rendezvous.
		We have for \((t,x)\)
		\begin{align*}
			(t(\ell), x(\ell)) &= \left(
			-\frac{s(\ell)}{1+v},\ 
			\frac{s(\ell)}{1+v}
			\right), \\
			(t(r), x(r)) 
			&= \left(
			t(\ell) + \frac{s(r) - vt(\ell) - x(\ell)}{1+v}
			,\ x(\ell) + \frac{s(r) - vt(\ell) - x(\ell)}{1+v}
			\right) \\
			&=
			\left(\frac{ (1+v)s(r) - 2s(\ell)}{(1+v)^2},\
			\frac{2s(\ell)v + (1+v)s(r)}{(1+v)^2}\right),
		\end{align*}
		and we have for \((t',x')\)
		\begin{align*}
		(t'(r), x'(r) &= \left(
		\frac{s(r)}{1+v},\ 
		\frac{s(r)}{1+v}
		\right), \\
		(t'(\ell), x'(\ell)) 
		&= \left(
		t'(r) +  \frac{x'(r) - (s(\ell) + vt'(r))  }{1+v},\
		x'(r) - \frac{x'(r) - (s(\ell) + vt'(r)) }{1+v}
		\right)\\
		&=
		\left(
		\frac{2 s(r) -(1+v)s(\ell)}{(1+v)^2},\
		\frac{2vs(r) + (1+v)s(\ell)}{(1+v)^2}
		\right).
	\end{align*}
    The difference in time of the rendezvous with the last client between both solutions is given by
	\begin{align*}
		t'(\ell)-t(r)
		{}&{}=
		\frac{2 s(r) -(1+v)s(\ell)}{(1+v)^2} - \frac{ (1+v)s(r) - 2s(\ell)}{(1+v)^2}	\\
		{}&{}=
		\frac{(1-v) ( s(\ell) + s(r))}{(1+v)^2}.
	\end{align*}
	The difference in distance to the origin of rendezvous with the last client between both solutions is given by
		\begin{align*}
		x(r) - |x'(\ell)|
		{}&{}= 
		\frac{2s(\ell)v + (1+v)s(r)}{(1+v)^2} + \frac{2vs(r) + (1+v)s(\ell)}{(1+v)^2}\\
		{}&{}= 
		\frac{(3 v + 1) (s(\ell) + s(r))}{(v + 1)^2}.
	\end{align*}
	Further, observe that \(C_{max}(t,x) = t(\ell) + |x(\ell)|\) and \(C_{max}(t',x') = t'(r) + |x'(r)|\). We have that 
		\begin{align*}
        C_{max}(t,x) - C_{max}(t',x')
        {}&{}= 
		x(r) - |x'(\ell)| - (t'(\ell)-t(r))\\
		{}&{}= 
		\frac{(3 v + 1) (s(\ell) + s(r))}{(v + 1)^2} - \frac{(1-v) ( s(\ell) + s(r))}{(1+v)^2}\\
		{}&{}= 
		\frac{4 v (s(\ell) + s(r))}{(v + 1)^2} \\
		{}&{}>
		0,
	\end{align*}
	using the assumption that \(s(r) > |s(\ell)|\) and \(0 < v < 1\). This implies that \(C_{max}(t',x') < C_{max}(t,x)\), leading to a contradiction and thereby showing the claim.

    We can now outline the algorithm. First, determine the client farthest from the origin. If this client is in \(L\), process \(L\) first, followed by \(R\); otherwise, process \(R\) first, followed by \(L\). By the previous claims, it follows that this algorithm returns an optimal solution and it is easy to see that it computes in \(O(n)\) time. 
		\begin{figure}[t!]
			\centering
			\begin{tikzpicture}[scale=1.5]
                    \coordinate (startUp) at (0,4); 
                    \coordinate (startLow) at (0,-2.5); 
                    \coordinate (startServer) at (0,0);  
                    \coordinate (UP1) at (3.2,3.2); 
                    \coordinate (UP2) at (7.12, -0.72); 
                    \coordinate (UPE) at (7.84,0); 
                    \coordinate (LP1) at (2, -2); 
                    \coordinate (LP2) at (6.4,2.4); 
                    \coordinate (LPE) at (8.8, 0); 
            
                    \draw[dashed] (startServer) -- (UP1) -- (UP2) -- (UPE); 
                    \draw[thick] (startServer) -- (LP1) -- (LP2) -- (LPE); 
                    \draw[dotted] (startUp) -- (LP2); 
                    \draw[dotted] (startLow) -- (UP2); 
            
                    \draw[thin, lightgray] (startServer) -- (LPE);
                    \node[left, black] at (startServer) {(0,0)};

                    \filldraw[fill=black] (UP1) circle (1pt) node[above, black] {\((t'(r), x'(r))\)};

                    \filldraw[fill=black] (LP2) circle (1pt) node[above, black] {\((t(r), x(r))\)};
                    
                    \filldraw[fill=black] (LP1) circle (1pt) node[below, black] {\((t(\ell), x(\ell))\)};

                    \filldraw[fill=black] (UP2) circle (1pt) node[below, black] {\((t'(\ell), x'(\ell))\)};
            

                \end{tikzpicture}
			\caption{Illustration of the argument corresponding to Theorem~\ref{thm:TSPNoProcSlow}. Depicted is a time-space diagram where the horizontal axis corresponds to time and the vertical axis corresponds to space. The solid line represents the trajectory of the server in solution \((t,x)\), the dashed line represents the trajectory of the server in solution \((t',x')\), and the dotted line represents the trajectory of client \(a\). The circles indicate the time-space pairs that are discussed in the argument. }
			\label{fig:proofLinearSlow}
		\end{figure}
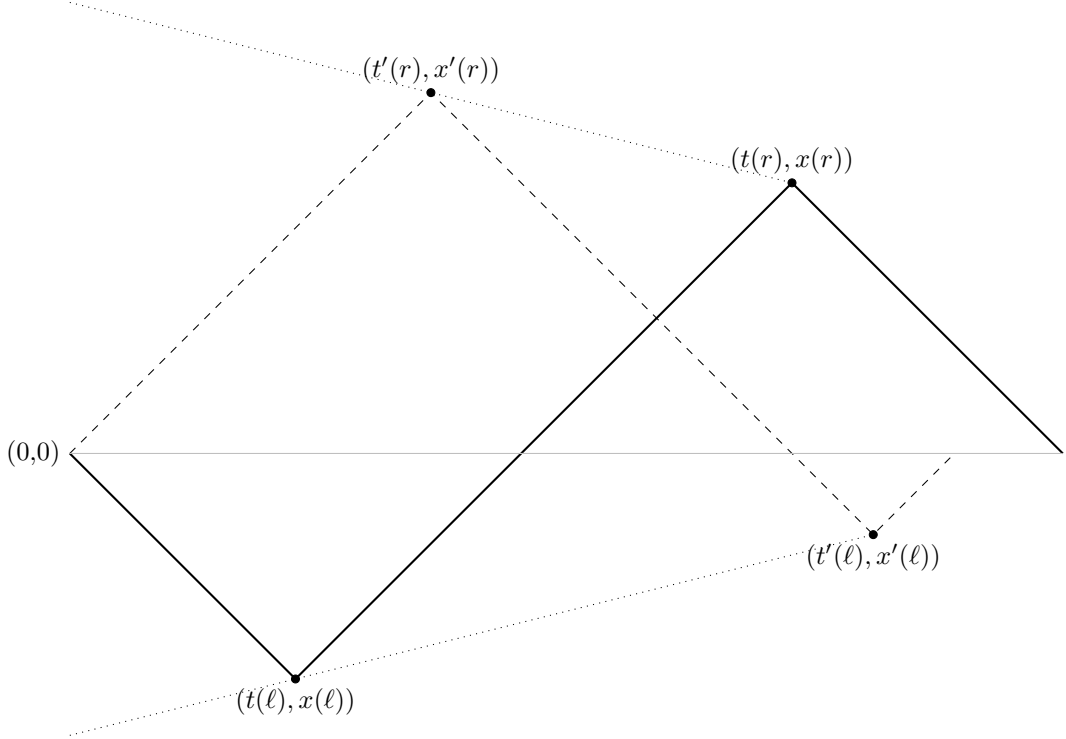
	\end{proof}
	
	\begin{theorem}\label{thm:TSPGenReleaseFast}
		If there are general processing times and only release times and if the clients are fast, then there exists an algorithm that solves the CLTSP in \(O(n)\) time.
	\end{theorem}
		\begin{proof}
		We claim that there exists a solution in which all rendezvous positions are at the origin.
		Consider an optimal solution \((t,x)\), where this is not the case and let \((t_j, x_j)\) denote the last rendezvous with \(x_j \neq 0\). 
		By Lemma~\ref{lemma:waitFreeClient}, we may assume that clients are colliding.
		Consider solution \((t',x')\) where the server moves back to the origin after rendezvous \(j-1\). 
        Let \(a = \sigma(j)\). 
        The time of the rendezvous with \(a\) is the maximum of the time when \(a\) arrives the origin and the time when the server arrives at the origin after the rendezvous with \(j-1\). 
        Thus, we define \((t',x')\) with \(\sigma' = \sigma\) and 
		\[
		(t'_i, x'_i) =
		\begin{cases}
			\left(\max\{c_{j-1} + |x_{j-1}|, t_j + |x_j|/v\}, 0\right), & \text{if } i = j, \\
			(t_i, x_i), & \text{otherwise.}
		\end{cases}
		\]
		In \((t,x)\), the time at which the client arrives at the origin, after processing \(a\) is given by \(T \coloneqq t_j + \tau(a) + |x_j|\).  In \((t',x')\), we know that the server arrives at the origin at time 
		\begin{align*}
			c_{j-1} + |x_j| 
			{}{}\leq 
		      c_{j-1} + |x_j| + |x_{j} - x_{j-1}|
			{}{}\leq
			t_j   + |x_j|
			{}{}= 
			T - \tau(a),
		\end{align*}
		using that \(c_{j-1} + |x_{j} - x_{j-1}| \leq t_j\). Further, we know that client \(a\) is at position \(x_j\) at time \(t_j\), thus \(a\) arrives at the origin at latest at time \(t_j + |x_j|/v \leq  T - \tau(a)\), using that \(v \geq 1\). Putting everything together, we know that
		\begin{align*}
			c'_j 
			{}{}=
			t'_j + \tau(a) 
			{}{}=
			\max\{t_j + \frac{|x_j|}{v}, c_{j-1} + |x_j|\} + \tau(a) 
			{}{}\leq 
			T.
		\end{align*}
		Consequently, we know that every rendezvous is reachable by the server and the clients and it directly follows that \(C_{max}(t',x') \leq C_{max}(t,x)\). Thus, \((t',x')\) is again an optimal solution and thereby proving the claim.
		
		It is then trivial to compute an optimal solution in \(O(n)\) time.
	\end{proof}
	
	\begin{theorem}\label{thm:TSPnoProcTimeWindow}
		If there are zero-processing times and time windows and if the clients are fast, then there exists an algorithm that solves the CLTSP in \(O(n\log n)\) time.
	\end{theorem}
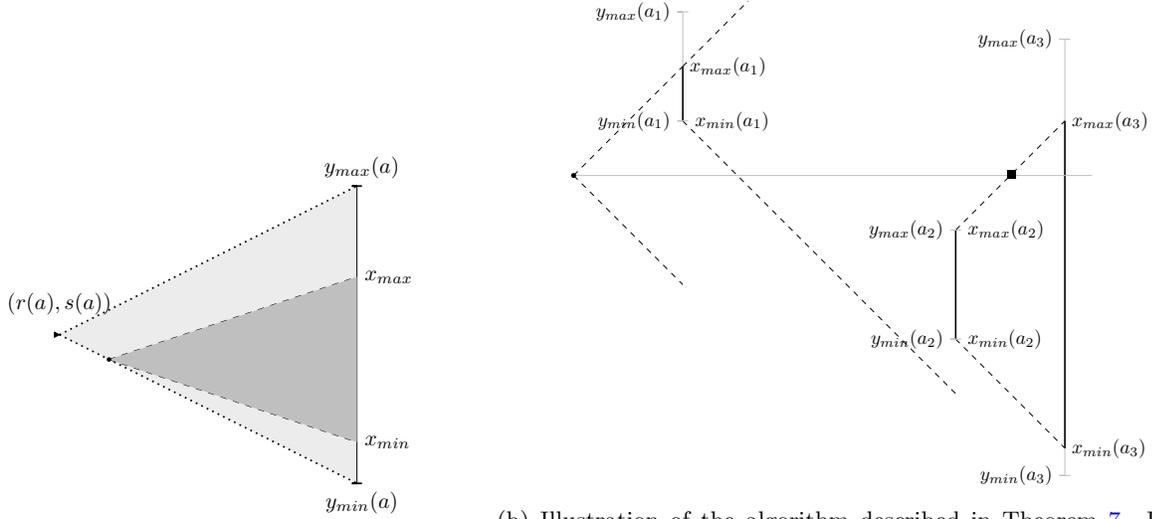
\begin{figure}[t!]
    \centering
    \begin{subfigure}[b]{0.35\textwidth}
        \centering
        \resizebox{1\textwidth}{!}{%
            \begin{tikzpicture}[scale=0.8]
    \coordinate (startClient) at (0,0); 
    \coordinate (ubClient) at (6,3); 
    \coordinate (lbClient) at (6,-3); 
    \coordinate (Rendezvous) at (1,-0.5); 
    \coordinate (ubServer) at (6,7/6); 
    \coordinate (lbServer) at (6,-13/6); 

    \fill[gray!15] (startClient) -- (ubClient) -- (lbClient) -- cycle;

    \draw[thick, dotted] (startClient) -- (ubClient); 
    \draw[thick, dotted] (startClient) -- (lbClient); 
    \draw (ubClient) -- (lbClient);
    \draw[dashed] (Rendezvous) -- (ubServer); 
    \draw[dashed] (Rendezvous) -- (lbServer); 
    \fill[gray!50] (Rendezvous) -- (ubServer) -- (lbServer) -- cycle;

    \node[right, black] at (ubServer) {\(x_{max}\)};
    \node[right, black] at (lbServer) {\(x_{min}\)};

    \filldraw[shift={(startClient)}, rotate=180] (0,0) -- (0.1cm,0.05cm) -- (0.1cm,-0.05cm) -- cycle 
    node[above=2mm] {\((r(a),s(a))\)};
    \draw[thick] (ubClient) -- ++(-0.1,0) -- ++(0.2,0) 
    node[above] {\(y_{max}(a)\)};
    \draw[thick] (lbClient) -- ++(-0.1,0) -- ++(0.2,0) 
    node[below] {\(y_{min}(a)\)};
    \filldraw (Rendezvous) circle (1pt);
\end{tikzpicture}
        }        
        \caption{
        Illustration of the argument in Theorem~\ref{thm:TSPnoProcTimeWindow}. The light-gray triangle represents the space-time points the client can reach, while the dark-grey triangle represents the space-time points the server can reach.}
        \label{fig:proofLinearFastBounds}
    \end{subfigure}
    \hfill
    \begin{subfigure}[b]{0.6\textwidth}
        \centering
        \resizebox{0.8\textwidth}{!}{%
            \begin{tikzpicture}[scale=1]
    \coordinate (startServer) at (0,0); 
    \coordinate (LB1) at (2,1); 
    \coordinate (RB1) at (2,3); 
    \coordinate (LB2) at (7,-3); 
    \coordinate (RB2) at (7,-1); 
    \coordinate (LB3) at (9,-5.5); 
    \coordinate (RB3) at (9,2.5); 
    
    \coordinate (SLB1) at (2,-2); 
    \coordinate (SRB1) at (2,2); 
    \coordinate (SLB2) at (7,-4); 
    \coordinate (SRB2) at (7,7); 
    \coordinate (SLB3) at (9,-5); 
    \coordinate (SRB3) at (9,1); 

    \coordinate (Cmax) at (8,0); 

    \draw[lightgray] (LB1) -- (RB1);
    \draw[lightgray] (LB2) -- (RB2);
    \draw[lightgray] (LB3) -- (RB3);
    \draw[thin, dashed] (startServer) -- (SLB1);
    \draw[thin, dashed] (startServer) -- (SRB1);
    \draw[thick] (LB1) -- (SRB1) node[right] {\(x_{max}(a_1)\)};
    \draw[thin, dashed] (SRB1) -- (3.2,3.2);
    \draw[thin, dashed] (LB1) -- (SLB2);
    \draw[thick] (RB2) -- (LB2);
    \draw[thin, dashed] (RB2) -- (SRB3);
    \draw[thin, dashed] (LB2) -- (SLB3);
    \draw[thick] (SRB3) node[right] {\(x_{max}(a_3)\)} -- (SLB3) node[right] {\(x_{min}(a_3)\)};
    \draw[thin, lightgray] (startServer) -- (9.5,0);

    \draw[lightgray] (RB1) -- ++(-0.1,0) -- ++(0.2,0) 
    node[left, black, xshift=-2mm] {\(y_{max}(a_1)\)};
    \draw[lightgray] (LB1) -- ++(-0.1,0) -- ++(0.2,0) 
    node[left, black, xshift=-2mm, yshift=0mm] {\(y_{min}(a_1)\)} node[right, black] {\(x_{min}(a_1)\)};
    \draw[lightgray] (RB2) -- ++(-0.1,0) -- ++(0.2,0) 
    node[left, black, xshift=-2mm] {\(y_{max}(a_2)\)} node[right, black] {\(x_{max}(a_2)\)};
    \draw[lightgray] (LB2) -- ++(-0.1,0) -- ++(0.2,0) 
    node[left, black, xshift=-2mm] {\(y_{min}(a_2)\)} node[right, black] {\(x_{min}(a_2)\)};
    \draw[lightgray] (RB3) -- ++(-0.1,0) -- ++(0.2,0) 
    node[left, black, xshift=-2mm] {\(y_{max}(a_3)\)};
    \draw[lightgray] (LB3) -- ++(-0.1,0) -- ++(0.2,0) 
    node[left, black, xshift=-2mm] {\(y_{min}(a_3)\)};
    \filldraw (startServer) circle (1pt);
    \fill (Cmax) + (-0.05, -0.05) rectangle ++(0.1, 0.1);
\end{tikzpicture}
        }
        
        \caption{Illustration of the algorithm described in Theorem~\ref{thm:TSPnoProcTimeWindow}.
        Depicted is an example with clients \(a_1, a_2\) and \(a_3\).
        The dotted lines are the bounds of the space-time points that the server can reach.
        The grey line segment corresponds to the space-time points that the respective client can reach by their deadlines.
        The black line segment corresponds to the space-time points that the respective client as well as the server can reach by their deadlines.}
        \label{fig:proofLinearFastAlgorithm}
    \end{subfigure}
    \caption{Illustration of the arguments corresponding to Theorem~\ref{thm:TSPnoProcTimeWindow}. Depicted are time-space diagram with time of the horizontal axis and time on the vertical axis.}
\end{figure}

	\begin{proof}
		For client \(a\), the positions reachable at time \(d(a)\) are defined by all \(x\) with \(y_{min}(a) \leq x \leq y_{max}(a)\), where \(y_{min}(a) \coloneqq s(a) - v(d(a) - r(a))\) and \(y_{max}(a) \coloneqq s(a) + v(d(a) - r(a))\).
		We claim that in a feasible solution, the server is at a position \(x\) at time \(d(a)\) with \(y_{min}(a) \leq x \leq  y_{max}(a)\). 
        This argument is illustrated in Figure~\ref{fig:proofLinearFastBounds}.
		To see this, consider a feasible solution \((t,x)\). 
        We know that the rendezvous with \(a\) satisfies \(r(a) \leq t(a) \leq d(a)\) and \(s(a) - v(t(a) - r(a)) \leq x(a) \leq s(a) + v(t(a) - r(a))\).
		All reachable positions by the server at time \(d(a)\) are given by \(x_{min} \coloneqq x(a) - (d(a) - t(a))\) and \(x_{max} \coloneqq x(a) + d(a) - t(a) \).
		Observe that
        \begin{align*}
		x_{min} - y_{min}(a) 
		{}&{}= 
		x(a) - (d(a) - t(a)) - (s(a) - v(d(a) - r(a)))\\
		{}&{}\geq
    	s(a) - v(t(a) - r(a)) - (d(a) - t(a)) - (s(a) - v(d(a) - r(a)))\\
		{}&{}\geq
		(v-1)(d(a)-t(a))\\
        {}&{}\geq
        0,
		\end{align*}
        using that  \(x(a) \geq s(a) - v(t(a) - r(a))\), \(v \geq 1\) and \(d(a) \geq t(a)\), implying that \(x_{min} \geq y_{min}(a)\). And we have
		\begin{align*}
			y_{max}(a) - x_{max} 
			{}&{}= 
			s(a) + v(d(a) - r(a)) - (x(a) + d(a) - t(a))\\
			{}&{}\geq
			s(a) + v(d(a) - r(a)) - (s(a) + v(t(a) - r(a)) + d(a) - t(a))\\
			{}&{}\geq
			(v-1)(d(a)-t(a)) \\
            {}&{}\geq\
            0,
		\end{align*}
		using that \(x(a) \leq s(a) + v(t(a) - r(a))\), \(v \geq 1\) and \(d(a) \geq t(a)\), implying that \(y_{max}(a) \geq x_{max}\) and the claim follows.
		
		Sort and index the clients in non-decreasing order of their deadlines, resulting in the sequence \(a_1, a_2,\ldots a_n\). 
		For a given \(a\), we define the reachable positions for \(a\) as the pair \((x_{min}(a), x_{max}(a))\), where there exists a feasible trajectory of the server that intersects a position \(x\) satisfying \(x_{min}(a) \leq x \leq x_{max}(a)\) at time \(d(a)\).
		For \(a_1\), we know that \(x_{min}(a_1) = \max\{y_{min}(a_1), -d(a_1) \}\) and \(x_{max}(a_1) = \min\{y_{max}(a_1), d(a_1) \}\). For \(a_i\) with \(i \geq 2\), we know that 
		 \(x_{min}(a_i) = \max\{y_{min}(a_i), x_{min}(a_{i-1}) - (d(a_i) - d(a_{i-1})) \}\) and \(x_{max}(a_i) = \min\{y_{max}(a_i), x_{max}(a_{i-1}) + (d(a_i) - d(a_{i-1})) \}\). 

		 We now describe the algorithm, which is illustrated in Figure~\ref{fig:proofLinearFastAlgorithm}. For \(i = 1, 2, \ldots, n\), compute pairs \((x_{min}(a_i) , x_{max}(a_i))\).  If there exists an \(i\) with \(x_{min}(a_i) > x_{max}(a_i)\),  return that the instance is infeasible.
         Otherwise, find the last client \(i\), such that the origin is not included in the interval \([x_{min}(a_i), x_{max}(a_i)]\), and let \(C \coloneqq d(a_i)  + \min\{ | x_{min}(a_i) |, | x_{max}(a_i) | \}\) if such \(i\) exists and  \(C \coloneqq 0\), otherwise. 
         Further, find the client \(j\) that arrives at the origin the latest and denote its arrival time at the origin with \(C'\). 
         Return \(\max\{C,C'\}\). 

         Based on the previous discussion, we know that the positions \(x_{min}(a_n) \leq x \leq x_{max}(a_n)\) are those positions that can be reached in any feasible solutions. 
         Thus, \(C\) represents the minimum time, that the server requires to serve clients \(a_1, \ldots, a_i\) within their time window and then return to the origin.
         Clients \(a_{i+1}, \ldots, a_n\) are then served at the origin in their time window.
         Observe that the solution is feasible.
         Correctness follows from a similar line of arguments as in Theorem~\ref{thm:TSPGenReleaseFast}.
         If \(a_n\) is not served at the origin, then, because clients are fast, the server reaches the origin earliest at the time that \(a_n\) arrives at the origin. 
         Thus, there must exist an optimal solution, where \(a_n\) is served at the origin.
         Repeating this procedure, shows that there must exist an optimal solution, where clients \(a_{i+1}, \ldots, a_n\) are served at the origin, showing the correctness of the algorithm.
         
		 Finally, observe that sorting the clients takes \(O(n\log n)\) time and computing the reachable positions takes \(O(n)\) time, thereby concluding the proof.
	\end{proof}

  \begin{figure}[t!]
			
		\end{figure}
	
	\subsection{Dynamic programming algorithms}\label{sec:DP}
	Before we discuss the algorithm, we provide a technical theorem establishing a dominance criterion that is used in the dynamic programming algorithms.

    \begin{lemma}\label{lemma:dominance}
        Consider the CLTSP, assuming general processing times, deadlines, and slow clients. 
        Let \(A' \subseteq A\) be a subset of clients, and let \(a \in A'\) be the client processed last among \(A'\). Define \( \mathcal{S} \) as the set of solutions that process clients in \(A'\) first and \(a\) last among them.
    
        If there exists an optimal solution \((t, x) \in \mathcal{S}\), then there also exists an optimal solution \((t^*, x^*) \in \mathcal{S}\) such that for all \((t', x') \in \mathcal{S}\), it holds that \(t^*(a) \leq t'(a)\).
	\end{lemma}
	\begin{proof}
    Assume that an optimal solution \((t, x) \in \mathcal{S}\) exists. If \(t(a)\) is already minimal among all solutions in \(\mathcal{S}\), the lemma is proven.
    Otherwise, suppose there exists a solution \((t', x') \in \mathcal{S}\) such that \(t'(a) < t(a)\). Consider modifying the original solution \((t, x)\) by replacing the rendezvous at client \(a\) with the earlier time \(t'(a)\) and adjusting the corresponding client position to maintain feasibility.
    From Lemma~\ref{lemma:waitFreeClient}, the server's position at the rendezvous with client \(a\) must satisfy \(s(a) = x'(a) + v t'(a)\). If this condition holds, and since \(t'(a) < t(a)\), the client’s position \(x'(a)\) must be greater than \(x(a)\) to ensure the meeting occurs at \(s(a)\).

    It can be shown that the server, starting from \(x'(a)\) at time \(t'(a)\), can reach the position \(x(a)\) by time
    \(
    t'(a) + (x'(a) - x(a))/v + \tau(a) < (s(a) - x(a))/{v} + \tau(a) \leq t(a) + \tau(a) = c(a),
    \)
    where \(0 < v < 1\). This guarantees that the rendezvous at \(x(a)\) can occur without delaying the overall schedule.

    Define a modified solution \((t^*, x^*)\) where
    \[
    (t^* (a'), x^*(a') ) = \begin{cases}
        (t'(a'), x'(a')) & \text{if } a' \in A \\ (t(a'), x(a')) & \text{if } a' \notin A
    \end{cases}
    \]
    This modified solution remains feasible since all subsequent rendezvous timings and positions remain unchanged, and \(t^*(a) \leq t(a) \leq d(a)\). Furthermore, it maintains the original makespan, ensuring optimality.
    Thus, there exists an optimal solution \((t^*, x^*) \in \mathcal{S}\) where \(t^*(a) \leq t'(a)\) for all \((t', x') \in \mathcal{S}\), concluding the proof.
	\end{proof}

	Consider the CLTSP assuming zero-processing times, only deadlines and slow clients.
	In each of the proofs of Lemmas~\ref{lemma:orderPreservingNoProc},~\ref{lemma:waitFreeServer}, and~\ref{lemma:waitFreeClient}, we show that for a feasible instance there exists an optimal solution with the respective property through a constructive argument that begins with an arbitrary optimal solution. 
	It follows that there exists an optimal solution that is order preserving with wait-free server and colliding clients. 
	We design a dynamic programming algorithm that computes such an optimal solution.
	The general idea of the algorithm is to construct a pendulum-like trajectory, where the direction changes are determined by alternating between processing clients in \(L\) and \(R\).
	Furthermore, the order-preserving property of the solution establishes a order for the clients in each set, while the wait-free and colliding properties allow us to compute the positions of the server and clients based on the makespan of a state.
	
	Let \(\ell(1), \ell(2), \ldots, \ell(n_L)\) represent the clients in \(L\), ordered by their distance from the origin in non-decreasing order. Similarly, let \(r(1), r(2), \ldots, r(n_R)\) represent the clients in \(R\), also ordered by their distance from the origin in non-decreasing order. We now construct the dynamic programming table. Define a state \(S = (i, j, d)\). Let \(D(S)\) represent the minimum time required to process the first \(i\) clients in \(L\) and the first \(j\) clients in \(R\), with the server's latest rendezvous was with a client in \(d \in \{L,R\}\). 
	
	We may assume feasibility for the first clients such that \(|s(\ell(1))|/(1+v) \leq d(\ell(1))\) and  \(|s(r(1))|/(1+v) \leq d(r(1))\), otherwise the instance is not feasible. Given the first clients, we can compute the time of their rendezvous and the state space is initiated with 
	\begin{equation}\label{eq:DP1Init}
		D(S) =
		\begin{cases}
			\frac{|s(\ell(1))|}{1+v}, & \text{if } S = (1, 0, L), \\
			\frac{|s(r(1))|}{1+v}, & \text{if } S = (0, 1, R).
		\end{cases}
	\end{equation}
	
	The dynamic programming table is updated by the following recursive relation. 
	We proceed in lexicographic order. Assume all states of the dynamic programming table up to but not including \(S\) are filled in.
	By construction, the preceding state must correspond to a rendezvous with either the client \(\ell(i-1)\) or \(r(j-1)\). 
	Thus, all potential states preceding \(S\) form a set
	\[
	prec(S) =
	\begin{cases}
		\{ (i-1, j, L), (i-1, j, R) \}, & \text{if } S = (i, j, L), \\
		\{ (i, j-1, L), (i, j-1, R) \}, & \text{if } S = (i, j, R).
	\end{cases}
	\]
	Consider a state \(S = (i,j,d)\) and a state \(S' = (i',j',d') \in prec(S)\).
    If \(D(S')=\infty\) we can skip the next steps and set the transition costs to \(T(S, S')=\infty\), otherwise the state \(D(S')\) represents the time at which the server processes either \(\ell(i')\) or \(r(j')\).
	Since the server and clients are wait-free, we can compute the server’s position as follows
	\[
	X(S') \coloneqq
	\begin{cases} 
		s(\ell(i')) + vD(S'), & \text{if } d' = L, \\
		s(r(j')) - vD(S'), & \text{otherwise.}
	\end{cases}
	\]
	Similarly, the position of the next client, either \(\ell(i)\) or \(r(j)\), at time \(D(S')\) can be computed as follows
	\[
	X(S, S') \coloneqq
	\begin{cases} 
		s(\ell(i)) + vD(S'), & \text{if } d = L, \\
		s(r(j)) - vD(S'), & \text{otherwise.}
	\end{cases}
	\]
	Therefore, the additional time of the rendezvous of \(S\), when transitioning to \(S\), is determined by the time required for the server and the client to close the distance between their current positions. Furthermore, if the arrival time exceeds the deadlines, we represent the state transition as infeasible by assigning it an additional time of infinity. Thus, let \(C = d(\ell(i))\) if \(d = L\) and \(C=d(r(j))\), otherwise. The additional time is then computed as follows
	\[
	T(S', S) = 
	\begin{cases} 
		\frac{|X(S') - X(S, S')|}{1 + v}, & \text{if } D(S') + \frac{|X(S') - X(S, S')|}{1 + v} \leq C,\\
		\infty, & \text{otherwise}.
	\end{cases}
	\]
	Then, the minimum time of processing from all clients to and including in state \(S\) is determined by 
	\begin{equation}\label{eq:DP1Recursion}
		D(S) = \min_{S' \in prec(S)} D(S') + T(S',S)
	\end{equation}
	Finally, to determine the makespan, we must calculate the minimum time required to process all clients, including the additional time needed for the server to return to the origin, which is given by
	\begin{equation}\label{eq:DP1Termination}
		\min \left\{ D(S_L) + \left| s(\ell(n_L)) + vD(S_L) \right|, \, D(S_R) + \left| s(r_R) - vD(S_R) \right| \right\},
	\end{equation}
	where \(S_L = (n_L, n_R, L)\) and \(S_R = (n_L, n_R, R)\). If \(D(S_L)\) and \(D(S_R)\) are both infinity, the algorithm returns {\sc Infeasible}.
	
	The dynamic programming algorithm is defined by its initialization in Equation~\eqref{eq:DP1Init} and its recursion in Equation~\eqref{eq:DP1Recursion}. The algorithm concludes by returning the minimum makespan as specified in Equation~\eqref{eq:DP1Termination} or returning {\sc Infeasible}. 
    We refer to Equations~\eqref{eq:DP1Init},~\eqref{eq:DP1Recursion} and~\eqref{eq:DP1Termination} as Algorithm 1.
	\begin{theorem}\label{thm:TSPNoProcDeadlines}
		If there are zero-processing times and there are only deadlines and if the clients are slow, 
		Algorithm 1 solves the CLTSP in \(O(n^2)\) time.
	\end{theorem}
	\begin{proof}
		We discuss the time complexity and the correctness of the algorithm separately.
		\medskip 
		
		\noindent \emph{Time complexity.} Note that the state space is limited by the \(O(n^2)\), while the computations in each state can be done in \(O(1)\) time. Therefore, the runtime complexity is \(O(n^2)\).
		
		\medskip
		\noindent \emph{Correctness.} 
        By lemma~\ref{lemma:dominance}, we know that a feasible instance there exists an optimal solution in which the time at each state in the dynamic programming table is minimized.
		Based on our discussion, it is clear that the dynamic programming algorithm enumerates all solutions that are order-preserving and with wait-free server and colliding clients, while pruning a path if a state can be reached in an earlier time by an alternative path.
		By Lemmas~\ref{lemma:orderPreservingNoProc},~\ref{lemma:waitFreeServer}, and~\ref{lemma:waitFreeClient}, we know that such an optimal solution must exist for a feasible instance and thus, the algorithm must return an optimal solution or {\sc Infeasible}, if such solution does not exist.
	\end{proof}
	
	Consider the CSTP with general processing times, no time constraints and slow clients. 
	Analogous to the previous problem variant, we know that an optimal solution exists that is order-preserving and wait-free server and colliding clients.
	We design a dynamic programming algorithm to compute such an optimal solution. 
	The general approach is similar to the previous algorithm, as we are again constructing a pendulum-like trajectory.
	In this problem variant, there are two ways for clients to reach their rendezvous positions. 
	A rendezvous with a client is referred to as \emph{waiting rendezvous}, meaning that the server is busy processing other clients upon the client’s arrival. 
	In this case, the client is processed as soon as the server finishes processing the clients that arrived earlier at that position.
	Otherwise, the rendezvous are referred to as \emph{wait-free rendezvous} when the client is processed immediately upon arrival. 
	The time of wait-free rendezvous is used to determine the position of the server, while the processing time of wait-free and waiting rendezvous is needed to compute the time at which the server departs from those positions.
	
	We now construct the dynamic programming table.
	Let \(S = (i, j, d)\), the state \(D(S)\) represents the following. 
	If \(i = n_L\) and \(j = n_R\), it denotes the time of processing all clients and returning to the origin, while the latest wait-free rendezvous was with a client in \(d\in \{L,R\}\).
	Otherwise, If \(d = L\), it denotes the time of a wait-free rendezvous with \(\ell(i)\), with all clients \(\ell(1), \ell(2) \ldots, \ell(i-1)\) and \(r(1), r(2) \ldots, r(j)\) already processed.
	If \(d = R\), it denotes the start of the time of a direct rendezvous with \(r(j)\), with all clients \(r(1), r(2) \ldots, r(j-1)\) and \(\ell(1), \ldots, \ell(i)\) already processed. 
	
	Given the first client, we can compute the time of their rendezvous and the state space is initiated with 
	\begin{equation}\label{eq:DP2Init}
		D(S) =
		\begin{cases}
			\frac{|s(\ell(1))|}{1+v}, & \text{if } S = (1, 0, L), \\
			\frac{|s(r(1))|}{1+v}, & \text{if } S = (0, 1, R).
		\end{cases}
	\end{equation}
	The dynamic programming table is updated by the following recursive relation. 
	We proceed in lexicographic order. 
	Assume all states of the dynamic programming table up to but not including \(S\) are filled in.
	By construction, the preceding state must correspond to a direct rendezvous with a client preceding either \(\ell(i)\) or \(r(j)\). Thus, all potential states preceding \(S\) form a set
		\[
	prec(S) =
	\begin{cases}
		\{ (i', j', d') \mid 0 \leq i' \leq i-1, \, 0 \leq j' \leq j, \, d' \in \{L,R\}\}, & \text{if } S = (i, j, L), \\
		\{ (i', j', d') \mid 0 \leq i' \leq i, \, 0 \leq j' \leq j-1, \, d' \in \{L,R\}\}, & \text{if } S = (i, j, R).
	\end{cases}
	\]
	
	Consider a state \(S\) and a state \(S'=(i', j', d') \in prec(S)\). 
	The state \(D(S')\) represents the time at which the server starts to processes either \(\ell(i')\) or \(r(j')\) in a direct rendezvous.
	Since the rendezvous is direct and server and clients are wait-free, we can compute the server's position as follows
	\[
	X(S') \coloneqq
	\begin{cases} 
		s(\ell(i')) + vD(S'), & \text{if } d' = L, \\
		s(r(j')) - vD(S'), & \text{otherwise.}
	\end{cases}
	\]
	At position \(X(S')\), the server processes client \(\ell(i')\) if \(d' = L\), and \(r(j')\) otherwise. While the server is busy, clients from either direction may arrive. We know that there exists an optimal solution in which the server processes each arriving client sequentially. This extends the time the server is busy, during which additional clients may arrive. Thus, we recursively determine the arriving clients until the server is no longer busy. Denote the last clients from each direction as \(\ell(i'')\) and \(r(j'')\) and the time at which the last client is finished with \(T(S')\). Thus, the server has processed clients \(\ell(1), \ell(2), \ldots, \ell(i'')\) and \(r(1), r(2), \ldots, r(j'')\) at time \(T(S')\) and is at position~\(X(S')\).
	
	Assume that \(S = (n_L, n_R, d')\) and \(i'' = n_L\) and \(j'' = n_R\). This represents the case where the server has processed all remaining clients at position \(X(S')\), and the state \(D(S)\) represents the final state. In this case, the server finishes processing at time \(T(S')\) and then returns to the origin, requiring an additional \(|X(S')|\) time units.
	
	Otherwise, assume that \(S = (i'' + 1, j'', L)\) or \(S = (i'', j''+1, R)\). 
	This represents the case, where the server has processed all clients up to \(i''\) and \(j''\), departs from position \(X(S')\) at time \(T(S')\) for a direct rendezvous with client \(\ell(i''+1)\) or \(r(j''+1)\). 
	At time \(T(S')\), the position of direct rendezvous is computed with
	\[
	X(S, S') \coloneqq
	\begin{cases} 
		s(\ell(i'' + 1)) + vT(S'), & \text{if } d = L, \\
		s(r(j'' + 1)) - vT(S'), & \text{otherwise.}
	\end{cases}
	\]
	Thus, the time of rendezvous of \(S\), when transitioning to state \(S'\), is given by 
	\[
	T(S, S') =
	\begin{cases}
		T(S') + |X(S')|, & \text{if } S = (n_L, n_R, d') \text{ and } i'' = n_L \text{ and } j'' = n_R  \\
		T(S') + \frac{|X(S') - X(S, S')|}{1+v}, & \text{else, if } S = (i'' + 1, j'', L) \text{ or } S = (i'', j'' + 1, R), \\
		\infty, & \text{otherwise.}
	\end{cases}
	\]
	Then, the minimum time of processing all clients up to and including state in \(S\) is determined by 
	\begin{equation}\label{eq:DP2Recursion}
		D(S) = \min_{S' \in prec(S)} T(S,S')
	\end{equation}
	Finally, the minimum makespan is computed with
	\begin{equation}\label{eq:DP2Termination}
		\min\{D(n_L, n_R, L), D(n_L, n_R, R)\}
	\end{equation}
	
	The dynamic programming algorithm is defined by its initialization in Equation~\eqref{eq:DP2Init} and its recursion in Equation~\eqref{eq:DP2Recursion}. The algorithm concludes by returning the minimum makespan as specified in Equation~\eqref{eq:DP2Termination}. We refer to Equations~\eqref{eq:DP2Init},~\eqref{eq:DP2Recursion} and~\eqref{eq:DP2Termination} as Algorithm 2. 
	\begin{theorem}\label{thm:TSPGenProcUnconstrained}
		If there are general processing times and there are no time constraints and if the clients are slow, 
		Algorithm 2 solves the CLTSP in \(O(n^3)\) time.
	\end{theorem}
	\begin{proof}
		We discuss the time complexity and the correctness of the algorithm separately.
		\medskip 
		
		\noindent \emph{Time complexity.} Note that the state space is bounded by \(O(n^2)\). 
		The while-loop is executed in each state and takes \(O(n)\) time.
		When the algorithm is implemented using forward iteration, there are at most two states reachable and these are returned by the while-loop. 
		All other state transitions correspond to infeasible solutions and can be pruned. 
		Thus, each state requires \(O(n)\) time. 
		It follows, that, with a reasonable implementation, the algorithm solves the problem in \(O(n^3)\) time.
		
		\medskip
		\noindent \emph{Correctness.} By lemma~\ref{lemma:dominance}, we know that there exists an optimal solution in which the time at each state in the dynamic programming table is minimized.
		Based on our discussion, it is clear that the dynamic programming algorithm enumerates all solutions that are order-preserving and with wait-free server and colliding clients, while pruning a path if a state can be reached in an earlier time by an alternative solution and if a state is not reachable.
		By Lemmas~\ref{lemma:orderPreservingGenProc},~\ref{lemma:waitFreeServer}, and~\ref{lemma:waitFreeClient}, we know that such an optimal solution must exist and thus, the algorithm must return an optimal solution.
	\end{proof}
	
	\section{Algorithms for the Line Traveling Repairman Problem with collaboration}\label{sec:algorithmsTRP}
	
	In this section, we develop an algorithm that computes an optimal solution for the CLTRP assuming zero-processing times and only deadlines and slow clients. Similar to the previous part, before discussing the algorithm, we give structural lemmas that show there exists an optimal solution with a specific structure. 
	
	The first two lemmas are analogues of Lemmas~\ref{lemma:orderPreservingNoProc},~\ref{lemma:waitFreeServer} and~\ref{lemma:waitFreeClient}. Their validity follows straightforwardly from similar arguments, and we omit their proofs.
	
	\begin{lemma}\label{lemma:orderPreservingTRP}
		If there are zero-processing times and only deadlines, then given an optimal solution \((t, x)\) for the CLTRP, we have that for any clients \(a, a'\in L\) or \(a, a'\in R\), such that \( |s(a)| < |s(a')|\), we have that \(t(a) < t(a')\). In other words, the solution is order-preserving.
	\end{lemma}
	\begin{lemma}\label{lemma:waitFreeServerTRP}
	If there are only deadlines and clients are slow, then given an optimal solution \((t, x)\) for the CLTRP, we have that \(t_{i} = c_{i-1} + |x_{i} - x_{i-1}|\) for all \(1 \leq i \leq n+1\). In other words, the server is wait-free.
\end{lemma}
	\begin{lemma}\label{lemma:waitFreeClientTRP}
		If there are only deadlines, given an optimal solution \((t, x)\) for the CLTRP for all \(a\in R\), we have that \(t(a) = s(a) - |s(a) - x(a)|/v\). For all \(a\in L\), we have that \( t(a) = s(a) + |s(a) - x(a)|/v \). In other words, the clients are colliding.
	\end{lemma}

	In the case of fast clients with zero processing times and only deadlines, it is straightforward to construct instances where it is optimal for the server to wait. However, for fast clients with no time constraints, we are not aware of any such examples. We conjecture that there exists an optimal solution in which the server is wait-free, but are not aware of a proof. We pose this question as an open problem for further research. If the conjecture holds true, the following dynamic programming algorithm can be extended for fast clients.
	
	\subsection{Algorithm}
	Consider the Traveling Repairman Problem with collaboration with zero-processing times and only deadlines. 
	From the previous Lemmas we know, that an optimal solution exists that is order-preserving and wait-free. 
	We design a dynamic programming algorithm to compute such an optimal solution. 
	The general approach is similar to the algorithms of the Traveling Salesman Problem with collaboration, as we again constructing a pendulum-like trajectory. 
	However, since the objective function involves the sum of completion times, it is not possible to determine the positions of the server and client solely based on a state value. 
	Therefore, it is necessary to include time explicitly in the state description.
	Let \(a_1, a_2, \ldots, a_n\) denote the clients, sorted in an order-preserving manner within \(L\) and \(R\), while alternating between clients from \(L\) and \(R\) whenever possible.
	We define the upper bound on the time that the server travels until he has reached all clients across all optimal solutions with 
	\[T \coloneqq |s(a_1)|  +  \sum_{i \geq 2} |s(a_{i}) - s(a_{i-1})| \leq \sum_{a\in A} 2|s(a)|.\]
	
	We now construct the dynamic programming table.
	Let \(S= (i,j,d,t)\), the state \(D(S)\) represents the minimum bound of the sum of completion times of all clients, while for \(d=L\) the latest direct rendezvous was at time \(t\) with client \(\ell(i)\), with all clients \(\ell(1), \ell(2), \ldots, \ell(i-1)\) and \(r(1), r(2), \ldots, r(j)\) already processed, and for \(d=R\) the latest rendezvous was at time \(t\) with client \(r(j)\), with all clients \(\ell(1), \ell(2), \ldots , \ell(i)\) and \(r(1), r(2), \ldots, r(j-1)\) already processed.
	
	We may assume that \(\frac{|s(\ell(1))|}{1+v} \leq d(\ell(1))\) and  \(\frac{|s(r(1))|}{1+v} \leq d(r(1))\). Given the first client, we can compute the time of their rendezvous and the state space is initiated with 
	\begin{equation}\label{eq:DP3Init}
		D(S) =
		\begin{cases}
			n\cdot \frac{ |s(\ell(1))|}{1+v}, & \text{if } S = (1, 0, L, \tfrac{|s(\ell(1))|}{1+v}), \\
			n\cdot \frac{|s(r(1))|}{1+v}, & \text{if } S = (0, 1, R, \tfrac{|s(r(1))|}{1+v}).
		\end{cases}
	\end{equation}
	The dynamic programming table is updated by the following recursive relation. 
	We proceed in lexicographic order. Assume all states of the dynamic programming table up to but not including \(S\) are filled in.
	By construction, the preceding state must correspond to a rendezvous with either client \(\ell(i-1)\) or \(r(j-1)\). Thus, all potential states preceding \(S\) form a set \[
	prec(S) =
	\begin{cases}
		\{ (i-1, j, d', t') \mid d' \in \{L,R\}, \, 0 \leq t' \leq t\}, & \text{if } S = (i, j, L), \\
		\{ (i, j-1, d', t') \mid d' \in \{L,R\}, \, 0 \leq t' \leq t\}, & \text{if } S = (i, j, R).
	\end{cases}
	\]
	Consider State \(S\) and a state \(S' = (i', j', d', t') \in {prec} (S)\). 
	The time at which the server processes either \(\ell(i')\) or \(r(j')\) is given by \(t'\). 
	Since the server and the clients are wait-free, we can compute the server's position as follows
	\[
	X(S') \coloneqq
	\begin{cases} 
		s(\ell(i')) + vt', & \text{if } d' = L, \\
		s(r(j')) - vt', & \text{otherwise.}
	\end{cases}
	\]
	Similarly, the position of the next direct client, either \(\ell(i)\) or \(r(j)\), at time \(t'\) can be computed as follows
	\[
	X(S, S') \coloneqq
	\begin{cases} 
		s(\ell(i)) + vt', & \text{if } d = L, \\
		s(r(j)) - vt', & \text{otherwise.}
	\end{cases}
	\]
	Let \(C = d(\ell(i))\) if \(d = L\) and \(C = d(r(j))\), otherwise. The increment on the state value is given by 
	\[
	T(S, S') = 
	\begin{cases}
		(n-i-j+1)(t-t'), & \text{if } t\le C \text{ and } t - t' = \frac{|X(S') - X(S,S')|}{1+v}, \\
		\infty, & \text{otherwise.}
	\end{cases}
	\]
	Then, the minimum bound of the sum of completion times for state \(S\) is determined by
	\begin{equation}\label{eq:DP3Recursion}
		D(S) = \min_{S' \in prec(S)} D(S') + T(S,S')
	\end{equation}
	Finally, the minimal sum of completion times for all clients is computed with
    \begin{equation}\label{eq:DP3Termination}
    \min  \{ D(n_L, n_R, d, t) \mid d\in \{L,R\}, \, 0 \leq t \leq T\}
    \end{equation}
	The dynamic programming algorithm is defined by its initialization in Equation~\eqref{eq:DP3Init} and its recursion in Equation~\eqref{eq:DP3Recursion}. The algorithm concludes by returning the minimum sum of completion times as specified in Equation~\eqref{eq:DP3Termination}. We refer to Equations~\eqref{eq:DP3Init},~\eqref{eq:DP3Recursion} and~\eqref{eq:DP3Termination} as Algorithm 3.
	
	\begin{theorem} \label{thm:TRP}
		If there are zero-processing times and there are only deadlines and clients are slow, the Algorithm 3 solves the CLTRP in \(O((Tn)^2)\) time.
	\end{theorem}
	\begin{proof}
		We discuss the time complexity and the correctness of the algorithm separately.
		\medskip 
		
		\noindent \emph{Time complexity.} Note that the state space is bounded by \(O(Tn^2)\). 
		When the algorithm is implemented using forward iteration, there are at most two states reachable, which can be computed in \(O(1)\) time.  All other state transitions correspond to infeasible solutions and can be pruned. 
        Further, observe that in Equation~\eqref{eq:DP3Termination}, we find the minimum across all values of \(T\).
		It follows, that, with a reasonable implementation, the algorithm solves the problem in \(O((Tn)^2)\) time.
		
		\medskip
		\noindent \emph{Correctness.} Based on our discussion, it is clear that the dynamic programming algorithm enumerates all solutions that are order-preserving and with wait-free server and colliding clients, while pruning a path if a state can be reached in an earlier time by an alternative solution and if a state is not reachable.
		By Lemmas~\ref{lemma:orderPreservingTRP},~\ref{lemma:waitFreeServerTRP}, and~\ref{lemma:waitFreeClientTRP}, we know that such an optimal solution must exist and thus, the algorithm must return an optimal solution.
	\end{proof}

	\section{Complexity lower bounds}\label{sec:hardness}
	In the preceding sections we focused on efficient solution algorithms and found that particularly if the clients are faster than the server ($v>1$), the presence of collaboration can make the problems considerably easier to solve than their counterparts without collaboration. In this section we will turn to complexity lower bounds on the hardness of the problem versions and first establish that in the case of slow clients ($0<v<1$) we can build on the complexity results already proven in the literature for the LTSP and LTRP and show that versions of the CLTSP and CLTRP cannot be easier to solve than their counterparts without collaboration as long as $v$ can be arbitrarily small. We will then show that the presence of collaboration can make both the CLTSP and CLTRP considerably harder to solve by proving that finding a feasible solution to the CLTSP with deadlines is NP-complete in the strong sense as long as $v>0$.

    Rather than formal reductions from LTSP and LTRP versions, we will first line out a general strategy for reductions in the case of $0<v<1$. Intuitively, it should be unsurprising that the presence of collaboration will matter less and less the slower the clients can move to assist the server. We can make this intuition more precise by first considering that there is an upper bound $K$ on the length of any reasonable tour of the server for any instance of LTSP and LTRP, which in the most general problem version can for instance be computed as $K=n(r^{max}+2 \cdot s^{max}+\tau^{max})$, where $r^{max}$ is the largest release date, $s^{max}$ is the largest absolute position and $\tau^{max}$ is the largest processing time of the problem instance. This follows immediately since the server will have served each client after at most $r^{max}+2 \cdot s^{max}+\tau^{max}$ time units (after waiting that the client is released, traveling to its location along the line and then processing the client) and the server will only deteriorate the objectives by waiting at any position. For specific problems this bound can of course be tightened considerably, but this will have no impact on the general argument. 
    Thus, results in Sitters \cite{sitters2004complexity} and Tsitsikilis \cite{tsitsiklis1992special} imply the following results.
    
    \begin{corollary}\label{cor:CLTRP}
        The CLTRP with slow clients, zero processing times, and release times is binary NP-hard.
    \end{corollary}
    
    \begin{corollary}\label{cor:CLTSP}
        The CLTSP with slow clients, general processing times, and release times is binary NP-hard.
    \end{corollary}

    Consider a feasible solution to an instance of LTSP and LTRP respectively with objective values of $F^LTSP$ and $F^LTRP$. Notice that we can construct a feasible solution to a corresponding instance of the CLTSP and CLTRP instance with the same objective values if the clients simply stay put on their starting positions and the server visits the clients in the same sequence. Of course by moving towards the server the clients can save the server some time, but if we set the velocity of the slow clients to $v^{CTSP}=\frac{1}{n \cdot K +1}$ in the case of the CTSP and $v^{CTRP}=\frac{1}{\frac{n(n+1)}{2} \cdot K +1}$ then it follows that as the server completes its roundtrip in the same sequence total savings have to be less than a single unit and thus $F^LTSP - 1 < F^CLTSP \le F^LTSP$ and $F^LTRP - 1 < F^CLTRP \le F^LTRP$. We now have a general strategy for reducing decision versions of the LTSP and LTRP that ask whether there exists a solution with an objective value smaller or equal than some constant $C$ to corresponding decision versions of CLTSP and CLTRP that ask the same question. If the velocity is set to a sufficiently small value then it follows immediately that there is a solution to the corresponding CLTSP or CLTRP instance if and only if there is such a solution for the LTSP or LTRP instance respectively. Notice that the size of $K$ is polynomially bounded in the largest number in the original LTSP or LTRP instance and thus the reduction can be carried out in polynomial time and can be used to infer NP-hardness in the strong as well as in the ordinary sense. We can thus transfer all known hardness results from the LTSP and LTRP literature to the collaborative cases with slow clients as long as the velocity is not strictly bounded from below. This of course leaves open the question of whether there CLTSP or CLTRP instances which become easier (or harder) for larger values of $v$ which could suitably be studied in future research. 

    In the following we rather seek to establish that collaboration can sometimes make finding even a feasible solution harder irrespective of the concrete velocity $v>0$. For this purpose we will study the feasibility version of the CLTSP/CLTRP problem with deadlines and prove the following:

	\begin{theorem}\label{thm:HardnessStrong}
		The feasibility problem with deadlines is NP-complete in the strong sense.
	\end{theorem} 
	
	\begin{proof}
		To prove that the feasibility problem with deadlines is NP-complete in the strong sense, we make use of a reduction from the 3-Partition Problem which is well known to be NP-complete in the strong sense, see Garey and Johnson \cite{garey2002computers}, and is stated as follows: Given \(3m\) integers \(A = \{a_1, a_2, \ldots, a_{3m}\}\) and a positive integer \(B\) such that \(\sum_{i=1}^{3m} a_i = mB\) and \(B/4 < a_i < B/2\) for all \(i\) is there a partition of \(A\) into \(m\) subsets \(A_1,A_2,\ldots, A_m\) such that it holds that \(\sum_{a \in A_j} a = B\) for \( j=1,\ldots,m\).
		
		We construct an instance of the feasibility problem with deadlines as follows. Introduce \(3m\) adjacent clients located at the origin with \(s_i = 0\) for \(i=1,\ldots,3m\), corresponding to the integers in \(A\) with a processing time of \(\tau_i = a_i\) for \(i=1,\ldots,3m\) and deadlines of \(d_i = m \cdot B\). Additionally, introduce \(m-1\) distant clients with a processing time of \(0\) and deadlines of \(d_{3m+j} = j \cdot B\) for \(j=1,\ldots,m-1\). 
        The velocity and distance of the distant clients are set such that \(|{s_{3m+j}}/{v}| = j \cdot B \) for \( j=1,\ldots,m-1\), which means that they arrive at position 0 after exactly \(1 \cdot B, 2 \cdot B, \ldots, (m-1) \cdot B\) time units if they move towards this point at full speed.
		
		As can be readily seen, a feasible solution requires that the server is free at time points \(1 \cdot B, 2 \cdot B, \ldots, (m-1) \cdot B\) in order to process the arriving distant clients which are due immediately. Given that the adjacent clients also need to be processed without any delay in order to service them all within \(3m\) time units, it follows that three adjacent clients with a total processing time of \(B\) need to be serviced in between any distant client which yields the required 3-partition. 
		
		As the problem is in NP, it follows that it is NP-completeness in the strong sense.
	\end{proof}
	
	\section{Conclusion}\label{sec:conclusion}
	In this work we extended the line Travelling Salesman and line Travelling Repairman problems by considering mobile clients that seek to cooperate with the server. We analyzed structural properties of several problem versions and identified efficient solution algorithms for specific problem versions and hardness results for others thereby also clarifying relationship between the original and collaborative problem versions. As could be seen, the presence of collaboration introduces interesting and new formal structures that can be exploited by dedicated solution algorithms.
    There are still some open questions with respect to important problem versions of the CLTSP but in particular the CLTRP, which could be tackled in future research projects, even though filling some of these gaps might require additonal insights into the problem structure of the standard LTRP. Further natural extensions of our research could investigate problem versions with more than one dimension of movement or multiple servers. 

	\paragraph{Acknowledgments.} Julian Golak received financial support from dtec.bw – the Digitalization and Technology Research Center of the Bundeswehr, which is funded by the European Union through the NextGenerationEU program.
	
	\bibliographystyle{abbrvnat}
	\bibliography{references.bib}

    \section*{Statements and Declarations}
\textbf{Competing Interests:} The authors declare that they have no competing interests.
\end{document}